\documentclass[11pt]{article}
\usepackage{amsfonts}
\usepackage{enumitem}
\usepackage{soul}
\usepackage{float}
\usepackage{latexsym}
\usepackage{nicefrac}
\usepackage{booktabs}
\usepackage[ruled]{algorithm2e}

\SetAlFnt{\small}
\SetAlCapFnt{\small}
\SetAlCapNameFnt{\small}
\SetAlCapHSkip{0pt}
\IncMargin{-\parindent}
\usepackage{mathtools}
\usepackage[
  bookmarks=true,
  bookmarksnumbered=true,
  bookmarksopen=true,
  pdfborder={0 0 0},
  breaklinks=true,
  colorlinks=true,
  linkcolor=black,
  citecolor=black,
  filecolor=black,
  urlcolor=black,
]{hyperref}
\usepackage{graphicx}
\usepackage{xcolor}
\usepackage[margin=1in]{geometry}
\usepackage{amssymb}
\usepackage{amsmath}
\usepackage[square]{natbib}
\setcitestyle{numbers}
\usepackage{amsthm}
\usepackage[capitalize]{cleveref}
\usepackage{bbm}
\usepackage{tikz}
\usetikzlibrary{decorations.pathreplacing,angles,quotes}
\theoremstyle{plain}
\newtheorem{theorem}{Theorem}[section]
\newtheorem{corollary}[theorem]{Corollary}
\newtheorem{lemma}[theorem]{Lemma}
\newtheorem{proposition}[theorem]{Proposition}
\newtheorem{observation}[theorem]{Observation}
\newtheorem{remark}[theorem]{Remark}

\newtheorem{fact}[theorem]{Fact}
\theoremstyle{definition}
\newtheorem{definition}[theorem]{Definition}

\DeclarePairedDelimiter\parentheses{(}{)}
\DeclarePairedDelimiter\braces{\{}{\}}
\DeclarePairedDelimiter\brackets{[}{]}
\DeclarePairedDelimiter\absolute{|}{|}
\DeclarePairedDelimiter\brackhalf{[}{)}
\DeclarePairedDelimiter\parhalf{(}{]}
\hypersetup{
  pdfauthor      = {Yakov Babichenko <yakovbab@technion.ac.il>, Inbal Talgam-Cohen <italgam@cs.technion.ac.il>, Haifeng Xu <mailto:hx4ad@virginia.edu>, Konstantin Zabarnyi <konstzab@gmail.com>},
  pdftitle       = {Regret-Minimizing Bayesian Persuasion}
}
\begin{document}
\newcommand{\opt}{{\mathrm{OPT}}}
\newcommand{\reg}{{\mathrm{Reg}}}
\newcommand{\apr}{{\mathrm{Apr}}}
\newcommand{\loss}{{\mathrm{Loss}}}
\newcommand{\argmax}{{\mathrm{argmax}}}
\newcommand{\supp}{{\mathrm{supp}}}
\newcommand{\at}{{\mathrm{at}}}
\newcommand{\vol}{{\mathrm{Vol}}}
\newcommand{\val}{{\mathrm{Val}}}
\newcommand{\arb}{{\mathrm{AR}}}
\newcommand{\mon}{{\mathrm{MON}}}
\newcommand{\mul}{{\mathrm{MON-MD}}}
\newcommand{\norm}[1]{\left\lVert#1\right\rVert}
\title{Regret-Minimizing Bayesian Persuasion}
\date{May 28, 2021}
\author{Yakov Babichenko\thanks{Technion--Israel Institute of Technology | \emph{E-mail}: \href{mailto:yakovbab@technion.ac.il}{yakovbab@technion.ac.il}. Yakov is supported by the Binational Science Foundation BSF grant no. 2018397 and by the German-Israeli Foundation for Scientific Research and Development GIF grant no. I-2526-407.6/2019.} 
\and 
Inbal Talgam-Cohen\thanks{Technion--Israel Institute of Technology | \emph{E-mail}: \href{mailto:italgam@cs.technion.ac.il}{italgam@cs.technion.ac.il}. Inbal is a Taub Fellow supported by the Taub Family Foundation.
This research was supported by the Israel Science Foundation grant no.~336/18.
} 
\and 
Haifeng Xu \thanks{University of Virginia | \emph{E-mail}: \href{mailto:hx4ad@virginia.edu}{hx4ad@virginia.edu}. Haifeng is supported by a Google Faculty Research Award.} \and Konstantin Zabarnyi\thanks{Technion--Israel Institute of Technology | \emph{E-mail}: \href{mailto:konstzab@gmail.com}{konstzab@gmail.com}.}}
\maketitle
\begin{abstract}
We study a Bayesian persuasion setting with binary actions (\emph{adopt} and \emph{reject}) for Receiver. We examine the following question -- how well can Sender perform, in terms of persuading Receiver to adopt, when ignorant of Receiver's utility? We take a robust (adversarial) approach to study this problem; that is, our goal is to design signaling schemes for Sender that perform well for \emph{all} possible Receiver's utilities. We measure performance of signaling schemes via the notion of (additive) \emph{regret}: the difference between Sender's hypothetically optimal utility had she known Receiver's utility function and her actual utility induced by the given scheme.

\vspace{\baselineskip}

On the negative side, we show that if Sender has no knowledge at all about Receiver's utility, then Sender has no signaling scheme that performs robustly well. On the positive side, we show that if Sender only knows Receiver's ordinal preferences of the states of nature -- i.e., Receiver's utility upon adoption is monotonic as a function of the state -- then Sender can guarantee a surprisingly low regret even when the number of states tends to infinity. In fact, we exactly pin down the minimum regret value that Sender can guarantee in this case, which turns out to be at most $1/e$. We further show that such positive results are not possible under the alternative performance measure of a multiplicative approximation ratio by proving that no constant ratio can be guaranteed even for monotonic Receiver's utility; this may serve to demonstrate the merits of regret as a robust performance measure that is not too pessimistic. Finally, we analyze an intermediate setting in between the no-knowledge and the ordinal-knowledge settings.
\end{abstract}
\section{Introduction}
\label{sec:intro}
\paragraph{{\bf Bayesian persuasion and applications.}} Since the seminal paper of~\citet{Kamenica2011}, the model of \emph{Bayesian persuasion} -- a.k.a.~\emph{information design} -- has been extensively studied in economics, computer science and operational research (for recent surveys see~\cite{Dughmi17,Kamenica2019,BergemannM19,Candogan20}). The model considers an informed \emph{Sender} who knows the \emph{state of nature}, and an uniformed \emph{Receiver} who does not know the state, but takes an action that affects both Receiver's and Sender's utilities. Sender has the ability to commit, before learning the state, to an information revelation policy called a \emph{signaling scheme}. The main question of interest is: what is the optimal Sender's utility, and what scheme should Sender choose to achieve this utility? The Bayesian persuasion model has quite a few applications in different domains, including interactions between: 
\begin{itemize}
    \item a prosecutor (Sender) and a judge (Receiver)~\cite{Kamenica2011}; the unknown state of nature is whether the defendant is guilty or not; the prosecutor's goal is to persuade the judge to convict;
    \item a seller of a product (Sender) and a potential buyer (Receiver) \cite{Kamenica2011,AB,Ozan19}; the unknown state is the product quality; the seller's goal is to persuade the buyer to purchase;
    \item a principle (Sender) delegating a project to an agent (Receiver) \cite[][]{DM}; the unknown state is the project quality; the principle's goal is to persuade the agent to maximize the effort invested in the project.
\end{itemize}
Further application domains of Bayesian persuasion include security~\cite{Xu15,Rabinovich15}, routing~\cite{BhaskarCKS16}, recommendation systems~\cite{Mansour2016bayesian}, auctions~\cite{emek2014signaling,bro2012send}, voting~\cite{Alonso14,cheng2015mixture}, queuing~\cite{Lingenbrink2019}, etc.
\paragraph{{\bf Knowledge of Receiver's utility.}} One of the fundamental assumptions underlying the Bayesian persuasion model is that Sender perfectly knows Receiver's utility, and she can use this detailed knowledge in her persuasion efforts. However, this assumption may be too demanding in some circumstances. In particular, in all examples above, it is reasonable to assume that Sender is uncertain about Receiver's precise utility. In the judge-prosecutor interaction, the prosecutor may be uncertain of the judge's exact utilities in the four cases of conviction vs.~acquittal, and guilty vs.~innocent defendant; thus, when designing her investigation (signaling scheme) that will convey to the judge a certain posterior probability of guilt/innocence, the prosecutor might not know the threshold probability above which the judge would convict. Similarly, in the seller-buyer interaction, the seller may be unaware of the buyer's utility as a function of the product's quality, and a principle may be unaware of her agent's utility (in particular -- because she might not know how costly the agent's effort is). Arguably, in almost all Bayesian persuasion applications, the scenario in which Sender is not fully aware of Receiver's utility may naturally arise.
\paragraph{{\bf Bayesian knowledge of Receiver's utility.}} These examples motivate the study of Bayesian persuasion in settings in which Sender has to decide how to reveal her private information despite not having complete knowledge of Receiver's utility. One approach for addressing this issue is \emph{Bayesian} -- Sender does not know Receiver's utility, but she has a prior distribution on it. However, relying on precise Bayesian knowledge raises the same brittleness issues as relying on knowledge of the utility itself, only ``one level up''. In the oft-cited words of Nobel Laureate Robert Wilson, ``I foresee
the progress of game theory as depending on successive reduction in the base of common knowledge required to conduct useful analyses of practical problems'' \cite{Wilson87} (see also \cite{Hayek45,Scarf58}).
As another consideration against the Bayesian route, in our setting it turns out that Bayesian modeling of uncertainty regarding the utility essentially reduces back to the standard persuasion model.\footnote{For each posterior distribution of Receiver over the states of nature, the \emph{distribution of actions} that Receiver would take is determined by Sender's prior on Receiver's utilities. This specifies the expected Sender's utility for each Receiver's posterior. Therefore, one can analyze the interaction using the standard Bayesian persuasion setting, and so this approach is not expected to shed new light on robust methods of persuasion.}
\paragraph{{\bf The robust approach.}}
An alternative way to address Sender's uncertainty is the \emph{robust approach} -- a.k.a.~the \emph{adversarial approach} or the \emph{prior-free approach}. The goal is to design signaling schemes that perform robustly well, that is -- perform well for \emph{all} Receiver's utilities, even when an adversary is allowed to choose the \emph{worst case} utility for Receiver. The adversarial lens has the potential to provide insights on issues such as: \emph{(a)~How harmful can it be for Sender to be unaware of Receiver's utility? (b)~What information about Receiver's utility is sufficient to enable ``reasonably good" persuasion? (c)~How should Sender approach persuading Receiver if Receiver's utility is unknown?} We shall return to these questions in Subsection~\ref{sub:res}.

There are several common robust approaches, which differ in their measure of performance and the benchmark they measure against. All three approaches have proved useful in appropriate contexts:
\begin{enumerate}
    \item {\bf Regret minimization.} The \emph{regret minimization approach} compares the performance of a signaling scheme when Sender \emph{does not know} the Receiver's utility function to the performance of an optimal scheme \emph{with knowledge} of Receiver's utility; the comparison is by considering the \emph{difference} between the optimal utilities in the two cases, known as the (additive) \emph{regret}. 
    The regret minimization approach is attributed to the classic work of \citet{Savage51} on decision theory.
    It is also the leading paradigm in online machine learning \cite{ShalevSBD14}. In economic contexts, two recent examples adopting this approach are \cite{ABS18} (in the context of information aggregation) and \cite{GuoS19} (in the context of monopoly regulation).
    \item {\bf Adversarial approximation.} The \emph{adversarial approximation approach} is similar to the previous one, but the comparison is by considering the corresponding \emph{ratio}. 
    Using the ratio to measure performance is the leading paradigm in approximation algorithms \cite{Vazirani03},
    partially due to it being scale-free.%
    \footnote{The approximation ratio does not depend on the units of measure and requires no normalization.}
    In economic contexts, the adversarial approximation approach underlies the well-developed line of research on prior-independent optimal auctions ~\cite{DhangwatnotaiRY15,Talgam20} and their sample complexity~\cite{ColeRoughgarden14} (see \cite{RoughgardenT19} for a survey).
    Another example is the work of~\cite{HurwiczS78} on sharecropping contracts.
    \item {\bf Minimax.} The \emph{minimax approach}%
    \footnote{Note that the name of this approach is slightly misleading, as all three robust approaches have a minimax flavor, and the difference is in what is being optimized within the minimax expression.}
    measures the \emph{absolute performance} of a signaling scheme with \emph{no knowledge} of Receiver's utility function; this approach has no benchmark.
    The minimax approach is attributed to the scholarship of \citet{Wald50}.
    There are many recent applications of the minimax approach in economics, 
    including to auctions~\cite{BandiB14,Carroll17,GravinL18,BeiGLT19} and contracts~\cite{Carroll15,DuttingRT19} (see the comprehensive survey of~\citet{Carroll19}).
\end{enumerate}
\paragraph{{\bf Our focus on regret.}}
In this paper, we mainly focus on the first approach of regret minimization. This approach is known to be less pessimistic than the minimax approach (see, e.g., \cite{Wiki21} for a concrete example). 
Indeed, \citet{Savage72} highlights the ``extreme pessimism'' of the minimax approach as a disadvantage relative to regret minimization (see Chapters 9.8 and 13.2).
In our particular setting, minimax seems to suggest that without further assumptions there is not much Sender can do (see more discussions in Section~\ref{sec:future}). As for the adversarial approximation approach, we study it in Section~\ref{sec:other-robustnesses} and show that it too is overly pessimistic. In light of these negative results, one take-away from our analysis is the merit of regret minimization as a tool for deriving economic insights, especially in settings in which the two other (arguably more popular) robust approaches fail to produce such insights. In this sense, our work can be seen as close to that of \citet{DworczakP20}, who also diverge from the standard minimax approach to better capture the best policy for Sender. Our motivation
is very similar to that of~\citet{castiglioni2020online}, whose interesting work studies Bayesian persuasion in an online learning framework with the goal of relaxing the assumption that Sender knows Receiver’s utility. We view our approach to the problem as complementary to theirs. In their model, Sender repeatedly faces Receiver with a non-binary action, whose type is chosen by Adversary at each round from a finite set of possibilities. Their regret notion is with respect to a best-in-hindsight single signaling scheme. In contrast, we minimize the regret over a single persuasion instance with respect to the best scheme tailored to that instance. We consider a binary-action Receiver with a quite general class of (continuum-many) possible utilities.
\subsection{Our Results}
\label{sub:res}
\paragraph{{\bf Settings of interest.}}
We consider a setting with $n$ states of nature, an arbitrary prior distribution over them and a binary-action Receiver whose possible actions are \emph{adoption} and \emph{rejection}. Sender aims to persuade Receiver to adopt. This simple persuasion setting is aligned with the motivating examples mentioned above (prosecutor-judge, seller-buyer and principle-agent). Note that a binary-action Receiver is a fundamental case in Bayesian persuasion studied in many works -- two recent examples are~\cite{KMZL,guo2019interval}.

We study the following three variants:
\begin{itemize}
    \item \emph{Arbitrary utilities}, with a completely ignorant Sender, who has no information whatsoever on Receiver's utility.
    \item \emph{Monotonic utilities}, with Sender knowing that Receiver's utility upon adoption is monotonic as a function of the state of nature. This assumption arises naturally when the state of nature reflects possible qualities of a certain product (higher quality yields higher Receiver's utility upon adoption).
    \item \emph{Multidimensional monotonic utilities}, with the state representing the qualities of several attributes (dimensions) of a product, each with a finite set of possible qualities. We make the natural assumption that Receiver's utility upon adoption is monotonic in each dimension. We focus on a constant number of attributes, as is natural when persuading to buy a certain product.
\end{itemize}
In each case, the problem is specified by the number of states $n$ and by the prior over these states. For some of these specifications our results are able to pin down the regret explicitly, and for others we provide asymptotic results for $n\to \infty$.
\paragraph{{\bf Results for arbitrary vs.~monotonic utilities.}}
Our adversarial analysis shows that Sender cannot hope for a nontrivial bound on her regret for arbitrary utilities with a large number of states (see Theorem~\ref{thm:arbitrary-asy}). In contrast, even if the number of states tends to infinity, the regret remains quite low (at most $\frac{1}{e}$) in the monotonic utility case (Theorem~\ref{thm:mon}).
These two results together provide answers to questions $(a)$ and $(b)$ posed when introducing the robust approach -- on the one hand, it might be very harmful for Sender to be unaware of Receiver's utility; on the other hand, knowing Receiver's ordinal preferences over the states may suffice for Sender to persuade Receiver ``reasonably well". Our results highlight monotonicity as the distinguishing property among settings in which Sender can persuade and those in which she should seek additional information on Receiver before approaching him.

Regarding question $(c)$ -- how should Sender persuade Receiver without knowing her utility -- our positive results are all constructive, i.e., in their proofs we provide a signaling scheme that achieves the claimed regret.

Interestingly, in many cases we discover that the regret-minimizing policies use \emph{infinitely many} signals. This is in sharp contrast to the standard persuasion model, in which binary-signal policies are sufficient. 

Another insight from our results for monotonic utilities is that the idea of ``pooling together'' the highest-utility states is useful not only in standard persuasion, in which such threshold policies are known to be optimal~\cite{RenaultSV17}, but also for regret minimization. It makes our work fall within the ``classic'' theme of the robust mechanism design literature, by which well-known mechanisms have robustness properties -- including linear contracts~\cite{Carroll15} and simultaneous ascending auctions~\cite{milgrom2000putting}. In the standard setting, the amount of probability mass on highest-utility states that Sender pools together in an optimal scheme is determined by Receiver's utility. In our setting, since Receiver's utility is unknown, this amount should be carefully randomized to minimize the regret. In fact, this amount is a non-negative real-valued random variable with a continuum-sized support; hence, it translates to a signaling scheme with a continuum of signals.
\paragraph{{\bf Further results.}}
In Subsection~\ref{sub:md}, we study multidimensional states of nature assuming that Receiver's utility is monotonic in every dimension. Each dimension represents an attribute of a certain product; the value of the corresponding coordinate of a state represents the quality level of that attribute. We focus on a constant number of attributes. We provide a positive result for product priors by upper-bounding the regret in terms of the number of attributes. However, as a corollary from the arbitrary utility analysis, the regret for general priors might approach $1$ as the number of possible quality levels grows to infinity, even for $2$ attributes.

In Section~\ref{sec:other-robustnesses}, we depart from regret minimization and consider the alternative adversarial approximation approach. We first establish that all our negative results (from the arbitrary and multidimensional monotonic utility cases) translate to this approach as well. Furthermore, we prove that even with monotonic utilities, Sender can robustly achieve only a logarithmic factor of the optimal utility upon knowing Receiver's utility.
\subsection{Techniques and Organization}
Robust persuasion can be seen as a zero-sum game between Sender and an adversary (see Section~\ref{sec:setting}). Since both players have a rich set of strategies (signaling schemes vs.~adversarial utility functions), it is unclear how to analyze this game to get a handle on its value. Our main positive result (Theorem~\ref{thm:mon}) uses the ordinal knowledge of states to reduce the complex game to a much simpler one in which each player simply chooses a threshold real number from which Receiver adopts. Without the ordinal knowledge, getting a positive result requires complex analysis; in Proposition~\ref{pro:arbitrary-ternary} we achieve this via geometric arguments and analytical geometry. As some of our proofs are technically involved, we defer many of them to the appendices and use the body of the paper for presenting the result statements, proof intuitions, and main takeaways. Section~\ref{sec:future} contains several interesting directions for future research.
\subsection{Additional Related Work}
\label{sub:related} 
Recently, there is a growing interest in understanding the Bayesian persuasion model from an adversarial (robust) perspective. This recent literature focuses on robustness with respect to different ingredients of the persuasion model. 
\cite{HuW20,Kosterina20,DworczakP20} study robustness of persuasion with respect to private information that Receiver might have.
\cite{IP20,LSZ,MPT,MOT,Ziegler20} consider persuasion with multiple receivers and study robustness with respect to the strategic behavior of the receivers after observing their signals. Most of the existing literature on robust Bayesian persuasion adopts the minimax approach, while we focus on regret minimization.
\section{Our Setting}
\label{sec:setting}
We start by describing the classic Bayesian persuasion setting with a single Sender and a single Receiver. We impose two restrictions on the standard model: (1)~Receiver has binary actions, which captures choosing whether to adopt a certain offer or not; (2)~Sender's utility is state-independent, which captures caring only about whether Receiver adopts the offer, regardless of the circumstances.
\paragraph{{\bf Prior, posterior and signaling scheme.}} Let $\Omega=\brackets*{n}=\braces*{1,...,n}$ be the \emph{state of nature space}, and let $\omega\in\Omega$ be the true \emph{state of nature}. Denote by $\Delta\parentheses*{\Omega}$ the set of all the probability distributions over $\Omega$. For every $q\in\Delta\parentheses*{\Omega}$, let $q_i$ be the probability the distribution $q$ assigns to the state of nature $i$. Define the \emph{support} of $q$ to be $\supp\parentheses*{q}:=\braces*{i\in\Omega: q_i>0}$. Let $\mu \in \Delta\parentheses*{\Omega}$ be the \emph{prior distribution} on $\Omega$ -- a.k.a.~the \emph{prior}. Assume that $\mu$ is publicly known. Furthermore, assume w.l.o.g.~that the probability of every state is strictly positive -- that is, $\mu_i>0$ for every $i\in\Omega$.\footnote{Otherwise, one can eliminate $i$ from $\Omega$.} Therefore, $\mu$ is an interior point of $\Delta\parentheses*{\Omega}$.

The \emph{signaling scheme} is a stochastic mapping $\pi:\Omega \to S$, where $S$ is a (finite or infinite) set of signals. Sender commits to signaling scheme $\pi$ and then observes the true state $\omega$; Receiver does not observe $\omega$. Upon learning $\omega$, Sender transmits to Receiver, according to $\pi\parentheses*{\omega}$, a \emph{signal realization} $s\in S$. After receiving $s$, Receiver updates his belief regarding the state of nature distribution to a \emph{posterior distribution} on $\Omega$ -- a.k.a.~the \emph{posterior} -- denoted by $p\parentheses*{s}\in\Delta\parentheses*{\Omega}$, where $p_i\parentheses*{s}=\Pr\brackets*{\omega=i\mid s}=\frac{\Pr_{s'\sim\pi\parentheses*{i}}\brackets*{s'=s}\mu_i}{\Pr\brackets*{s}}$. Only the posterior distribution is important to the persuasion instance outcome, rather than the signal realization itself.

We slightly abuse the notation and use $\pi$ also to denote the distribution over the elements of $S$ induced by the signaling scheme $\pi$ considering the prior $\mu$. That is, for every $s_0\in S$:
\begin{align*}
    &\Pr_{s\sim\pi}\brackets*{s=s_0}=\sum_{i\in\Omega}\mu_i\cdot\Pr_{s'\sim\pi\parentheses*{i}}\brackets*{s'=s_0}.
\end{align*}
\begin{remark}
\label{rem:plausibility}
It is well-known (see, e.g.,~\cite{Blackwell53,AM}) that a distribution over Receiver's posteriors $p\in \Delta\parentheses*{\Omega}$ is implementable by some signaling scheme $\pi$ if and only if $\mathbb{E}\brackets*{p}=\mu$. We refer to this condition as \emph{Bayes-plausibility}.
\end{remark}
\paragraph{{\bf Utilities and adoption.}} In our setting, Receiver has a binary \emph{action space}; that is, Receiver's \emph{action} $a$ is selected from $\braces*{0,1}$. Call action~$1$ \emph{adoption} and action~$0$ \emph{rejection}. Sender's utility is a function $u_s:\braces*{0,1}\to\brackets*{0,1}$ of Receiver's action. We assume $u_s\parentheses*{a}:=a$, i.e., Sender wants Receiver to take action $1$ (to adopt). Thus, the expected Sender's utility equals the probability that Receiver adopts.

Receiver's utility is a function $u_r:\Omega \times \braces*{0,1} \to \mathbb{R}$ of the state of nature and Receiver's action. W.l.o.g., we normalize Receiver's utility $u_r\parentheses*{i,0}$ to zero for every $i\in \Omega$. That is, \emph{Receiver's utility for choosing to reject is always $0$}, regardless of the state of nature.

Given the posterior $p\parentheses*{s}$, Receiver adopts if and only if:\footnote{We assume that ties are broken in Sender's favour, as is standard in the Bayesian persuasion literature.}
\begin{equation}
\label{eq:adoption}
\mathbb{E}_{\omega'\sim p\parentheses*{s}}\brackets*{u_r\parentheses*{\omega',1}}\geq 0.
\end{equation}
\begin{definition}
The \emph{adoption region} $A=A\parentheses*{u_r}\subseteq\Delta\parentheses*{\Omega}$ is the set of posteriors that lead to Receiver's adoption: $A:=\braces*{p\in \Delta(\Omega):\mathbb{E}_{\omega'\sim p}\brackets*{u_r\parentheses*{\omega',1}}\geq 0}$.
\end{definition}
For a signaling scheme $\pi$ and Receiver's utility function $u_r$, denote Sender's expected utility over $\pi$ by:
\begin{equation*}
u\parentheses*{\pi,u_r} :=\Pr_{s\sim\pi}\brackets*{p\parentheses*{s}\in A\parentheses*{u_r}}.
\end{equation*}
\paragraph{\bf Objective.}
Given a Bayesian persuasion setting with (known) $u_r$ as Receiver's utility, the standard goal of the designer (Sender) is to compute a signaling scheme~$\pi^*=\pi^*\parentheses*{u_r}$ that maximizes Sender's expected utility: $\pi^*\in\arg\max_{\pi}\braces*{u\parentheses*{\pi,u_r}}$. In our binary-action setting, computing the optimum $\pi^*$ is a well-understood problem. One general approach to tackle the problem is by first expressing $u_s$ as a function of the posterior, then taking the concavification and evaluating it at the point representing the prior \citep[][]{Kamenica2011}.\footnote{Note that $u_s$ equals the indicator function $\textbf{1}_{p\parentheses*{s}\in A}$; since $A$ is a half-space, the concavification of $u_s$ can be computed and characterized.} 
\subsection{Threshold Signaling Schemes}
We now introduce a class of signaling schemes that can be used to optimally solve the setting with a binary-action Receiver in the standard model with Receiver's utility function known to Sender. Interestingly, this class plays a key part in our results for the adversarial model with unknown utility. We begin by describing the \emph{knapsack method} for solving a persuasion problem; then we formulate the class of signaling schemes that arise from this approach.
\paragraph{\bf The knapsack method.}
An alternative to the general concavification method is a greedy approach that reduces our problem to a fractional \emph{knapsack} instance. States are treated as the knapsack items, and their prior probabilities according to $\mu$ are the weights. Recall that every state induces utility $u_r\parentheses*{i,1}$ to Receiver upon adoption. Sender's goal is to add to the knapsack a maximum-weight set of states (possibly -- their fractions), while keeping Receiver's expected utility non-negative (the expectation is over a random state drawn w.p.~proportional to its weight in the knapsack). The knapsack instance can be solved greedily by sorting the states $i\in \Omega$ according to $u_r\parentheses*{i,1}$ in a non-increasing order, and then continuously adding masses of states to the knapsack as long as the expected Receiver's utility stays non-negative. Such an approach has been adopted by~\cite{RenaultSV17}.

An equivalent way to present the knapsack approach is as follows. Let $i_1,...,i_n$ be an ordering of the states s.t.~$u_r\parentheses*{i_1,1}\leq u_r\parentheses*{i_2,1}\leq...\leq u_r\parentheses*{i_n,1}$. One can interpret the state space equipped with the prior distribution as drawing uniformly a real number in $\brackets*{0,1}$ -- called the \emph{real-valued state} -- s.t.~for every $m\in \brackets*{n}$, all the realizations in the segment $\parhalf*{\sum_{l<m}\mu_{i_l},\sum_{l\leq m}\mu_{i_l}}$ correspond to the state $i_m$.\footnote{To demonstrate this, consider $n=4$ states $i_1,i_2,i_3,i_4$ with prior $\mu\parentheses*{i_1,i_2,i_3,i_4}=\parentheses*{0.2,0.3,0.1,0.4}$. Then, e.g., if the real-valued state is $0.55$, the corresponding state from $\Omega$ is $i_3$.} Now we define the notion of a \emph{threshold signaling scheme}.
\begin{definition}
\label{def:threshold}
For every $t\in \brackets*{0,1}$, the \emph{$t$-threshold signaling scheme} is a binary signaling scheme that reveals whether the real-valued state is below $t$ or not.

Equivalently, for $t\in \parhalf*{0,1}$, in the state space $\Omega$, let $j\in \brackets*{n}$ be s.t.~$t\in \parhalf*{\sum_{l<j}\mu_{i_l},\sum_{l\leq j}\mu_{i_l}}$. Then the \emph{$t$-threshold signaling scheme} is a binary-signal scheme with a \emph{high} and a \emph{low} signal, s.t.~the high signal is sent w.p.~$1$ if $\omega > i_j$, w.p.~$\frac{\sum_{l\leq j}\mu_{i_l}-t}{\mu_{i_j}}$ if $\omega=i_j$ and w.p.~$0$ if $\omega<i_j$. For $t=0$, the $0$-threshold signaling scheme is the scheme that reveals no information.
\end{definition}
Given a threshold signaling scheme, we say of the states for which it sends the high (resp., low) signal that they are \emph{pooled together}. The knapsack approach is based on the following fact:
\begin{fact}[See, e.g.,~\cite{RenaultSV17}]
\label{fact:optimal-threshold}
Every persuasion problem with a binary Receiver's action space admits an optimal $x$-threshold signaling scheme for some $x$. Moreover, for every $y\geq x$, Receiver adopts after observing the high signal in the $y$-threshold signaling scheme.
\end{fact}
\subsection{Robust Approach}
\label{sub:adversarial}
We depart from the standard model by assuming from now on that Receiver's utility function $u_r$ is \emph{unknown} to Sender. Besides Sender and Receiver, we introduce a third entity called \emph{Adversary}. Given $\mu$ and $\pi$, Adversary aims to set Receiver's utility $u_r$ in a way that makes the performance of $\pi$ as bad as possible. Thus, Sender aims to design a signaling scheme with a worst-case guarantee -- it should perform well for \emph{all} possible Receiver's utility functions~$u_r$.
\paragraph{{\bf Regret definition.}}
Here we formalize the (additive) regret minimization setting on which we mostly focus in this paper, as mentioned in Section~\ref{sec:intro}. Fix a prior distribution $\mu$. Consider $\pi^*\parentheses*{u_r}$, a signaling scheme that maximizes Sender's expected utility in the standard persuasion setting in which Receiver's utility $u_r$ is known; denote by $u^*\parentheses*{u_r}$ the expected Sender's utility $\pi^*\parentheses*{u_r}$ yields. In our adversarial setting, given any signaling scheme $\pi$ ignorant to $u_r$, Adversary aims to set $u_r$ in a way that maximizes Sender's regret -- that is, maximizes the difference between what Sender could have gotten had she known $u_r$ and what she gets with $\pi$. In other words, Adversary tries to punish Sender for not choosing $\pi^*\parentheses*{u_r}$ as much as possible. The regret of $\pi$ is, therefore, defined as:
\begin{definition}
\label{def:reg_pi}
The \emph{regret of a signaling scheme $\pi$} is
   $\reg\parentheses*{\pi}:=\sup_{u_r} \braces*{u^*\parentheses*{u_r}-u\parentheses*{\pi,u_r}}.$
\end{definition}
Sender aims to set $\pi$ in a way that minimizes the regret; thus, the regret of the setting is defined as follows:
\begin{definition}
\label{def:reg}
The \emph{regret of the setting} is   $\reg:=\inf_{\pi}\reg\parentheses*{\pi}.$ 
\end{definition}
\paragraph{{\bf The scale of regret.}}
Note that we always have $0\leq \reg \leq 1$. Consider the following possibilities.
\begin{enumerate}
    \item $\reg=0$, i.e., Sender can ensure utility $u^*\parentheses*{u_r}$ without knowing $u_r$. Such a result is ``too good to be true'', as indicated by our results.
    \item $0\ll \reg \ll 1$ even when the problem becomes large (the number of states grows to infinity), i.e., the additive loss from not knowing Receiver's utility is strictly less than $1$. Such a result can be viewed as positive. 
    \item $\reg \to 1$ as the problem becomes large; that is, in the worst case, Sender might lose her entire utility from not knowing $u_r$, as her expected utility might approach $0$ while she could have gotten utility approaching $1$ upon knowing $u_r$. Such a result is negative. One can also study the rate of convergence of $\reg$ to $1$ as a measure of how negative the result is.
\end{enumerate}
\paragraph{{\bf Zero-sum game perspective.}} Our setting can be analyzed via the following two-player zero-sum game $G$. The players are Adversary and Sender. Adversary's pure strategies are mixtures over the functions $\braces*{u_r}$, while Sender's pure strategies are the signaling schemes $\braces*{\pi}$.\footnote{Note that a mixture over signaling schemes is a signaling scheme.} We remark that our proofs do \emph{not} rely on the fact that $G$ has a value. Indeed, to bound $\reg$ from below by some bound $v$, we either describe explicitly a mixed strategy of Adversary ensuring $\reg\parentheses*{\pi}\geq v$ for every $\pi$, or show that for every choice of $\pi$ by Sender, Adversary can ensure a regret of $v$. To bound $\reg$ from above by some bound $v$, we explicitly describe a signaling scheme $\pi$ s.t.~$\reg\parentheses*{\pi}\leq v$. However, for completeness, we mention that Sion's Minimax Theorem~\cite{Sion} holds in our setting.\footnote{Sion's Minimax Theorem holds since the objective function is linear and lower semi-continuous as a function of $\pi$ (w.r.t.~the Lévy–Prokhorov metric on $\Delta\parentheses*{\Omega}$), and is linear and continuous as a function of the mixture over Receiver's utility functions.} Therefore, while we define $\reg$ as $\inf_\pi \sup_{u_r}\braces*{u^*\parentheses*{u_r}-u\parentheses*{\pi,u_r}}=\inf_\pi \sup_{\Delta\parentheses*{u_r}} \braces*{u^*\parentheses*{u_r}-u\parentheses*{\pi,u_r}}$, it also equals $\sup_{\Delta\parentheses*{u_r}} \inf_{\pi} \braces*{u^*\parentheses*{u_r}-u\parentheses*{\pi,u_r}}$. That is, if Sender can ensure a regret of at most $v$ after observing Adversary's strategy, then she has a signaling scheme that ensures $v$ even without knowing Adversary's strategy.
\section{Regret Minimization}
\label{sec:results}
We consider three cases as described in Section~\ref{sec:intro}: first, the case in which Sender has no information at all about Receiver's utility (Subsection~\ref{sub:ar}); second, the case in which Sender knows that Receiver's utility upon adoption is monotonic as a function of the state (Subsection~\ref{sub:mon}); finally, a case with a multidimensional state space and a Sender who knows that Receiver's utility is monotonic in each dimension (Subsection~\ref{sub:md}). As mentioned in Subsection~\ref{sub:res}, in the first case we obtain a negative asymptotic result, while in the second case we have a positive result; in the third case, we obtain a positive result for product priors, but a negative asymptotic result similar to the one from the first case holds for general priors.
\subsection{Arbitrary Utilities}
\label{sub:ar}
Denote by $\reg_{\arb}$ the regret for arbitrary utilities. That is, in Definition~\ref{def:reg_pi}, the supremum is taken over all Receiver's utilities. In this section, we start by characterizing the minimum regret for $n=2$ states, as well as for $n=3$ states with a uniform prior $\mu$ -- in both cases, $\reg_{\arb}=\frac{1}{2}$. Unfortunately, such constant regret cannot be expected in general for arbitrary utilities. Indeed, our main result of this section shows that for any $n\geq 2$, the regret satisfies $1-\frac{2}{\sqrt{n}}\leq \reg_{\arb} \leq 1-\frac{1}{4n^2}$. In particular, $\reg_{\arb}\to_{n\to\infty} 1$.
\begin{observation}
\label{obs:arb_adoption}
With arbitrary utilities, the set of all possible choices of $A\parentheses*{u_r}$ by Adversary coincides with the set of all the polytopes obtained by cutting $\Delta\parentheses*{\Omega}$ by a hyperplane.
\end{observation}
This observation follows from Equation~\ref{eq:adoption}, as for a fixed $u_r$, the condition $\mathbb{E}_{\omega'\sim p\parentheses*{s}} \brackets*{u_r\parentheses*{\omega',1}}=0$ specifies a hyperplane in $\Delta\parentheses*{\Omega}$ (see Figure~\ref{fig:adoption} for illustration).
\begin{figure}[H]
    \centering
    \caption{An illustration of the adoption region for arbitrary utilities and $n=3$.}
    \label{fig:adoption}
    \begin{tikzpicture}[scale=2]
        \draw (0,0)--(1,0)--(0.5,0.86)--(0,0);
        \draw (0.475,-0.215)--(0.1,0.86);
        \filldraw[gray] (0,0)--(0.4,0)--(0.25,0.43);
        \node at (0.2,0.15) {$A$};
        \node[left] at (0.2,0.6) {$\mathbb{E}_{\omega' \sim p\parentheses*{s}}\brackets*{u_r\parentheses*{\omega',1}}=0$};
        \node[below] at (0,0) {$\omega=1$};
        \node[below] at (1,0) {$\omega=2$};
        \node[above] at (0.5,0.86) {$\omega=3$};
    \end{tikzpicture}
\end{figure}
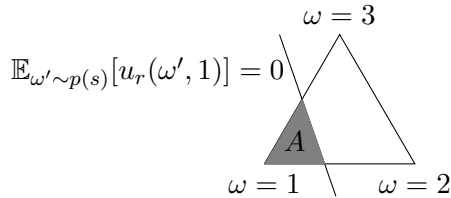
For a binary state space we have the following result.
\begin{proposition}
\label{pro:arbitrary-binary}
For $n=2$ and any prior $\mu$: $\reg_{\arb}=\frac{1}{2}$.
\end{proposition}
\begin{proof}
Let us show first that $\reg_{\arb}\leq\frac{1}{2}$. For simplicity of notation, we identify a prior/posterior with the probability it specifies for $\omega=1$. Assume w.l.o.g.~that $\mu_1\leq \frac{1}{2}$. Consider a binary-signal scheme $\pi$ that induces the posteriors $0$ and $2\mu_1$ with equal probability of $\frac{1}{2}$.\footnote{$\pi$ is Bayes-plausible, as the expected posterior probability of $\omega=1$ equals $\mu_1$; therefore, by Remark~\ref{rem:plausibility}, it indeed specifies a signaling scheme.} If $\brackets*{0,2\mu_1}\cap A\neq\emptyset$, then either $0\in A$ or $2\mu_1\in A$ (or both); thus, the adoption probability upon $\pi$ is at least $\frac{1}{2}$, and $\reg\parentheses*{\pi}\leq 1-\frac{1}{2}=\frac{1}{2}$. Otherwise, $A$ has the form $\brackets*{t,1}$ for some $t>2\mu_1$. By Markov's inequality, $\Pr_{s\sim\pi^*}\brackets*{p\parentheses*{s}\in A}\leq\frac{\mu_1}{t}<\frac{1}{2}$, and again $\reg\parentheses*{\pi}\leq\frac{1}{2}$, as desired. Therefore, the regret is at most $\frac{1}{2}$ for every fixed $u_r$; thus, it also holds for every mixture over them.

Conversely, let us show that for every signaling scheme $\pi$, Adversary can pick a Receiver's utility function $u_r$ -- or equivalently, an adoption region $A$ -- s.t.~the regret would be at least $\frac{1}{2}$.

Indeed, fix $\pi$. If $\Pr_{s\sim\pi} \brackets*{p\parentheses*{s}\in \brackets*{0,\mu_1}}\leq \frac{1}{2}$, then Adversary can set $A=\brackets*{0,\mu_1}$: on the one hand, $u^*=1$ is achieved (if Sender knows $A$) by the no-information scheme in which $p\parentheses*{s}=\mu_1$ w.p.~$1$, as $\mu_1\in A$; on the other hand, $u\parentheses*{\pi,u_r}=\Pr_{s\sim\pi} \brackets*{p\parentheses*{s}\in A}=\Pr_{s\sim\pi} \brackets*{p\parentheses*{s}\in \brackets*{0,\mu_1}}\leq\frac{1}{2}$. Hence, the regret is at least $\frac{1}{2}$, as desired.

It remains to consider the case in which $\Pr_{s\sim\pi} \brackets*{p\parentheses*{s}\in\brackets*{0,\mu_1}}>\frac{1}{2}$.

Set $\epsilon:=\frac{1}{2}-\Pr_{s\sim\pi}\brackets*{p\parentheses*{s}\in \parhalf{\mu_1,1}}\in\parhalf*{0,\frac{1}{2}}$ and let Adversary take $A=\brackets*{\frac{\mu_1}{1-\epsilon},1}$.\footnote{Note that $0<\mu_1<\frac{\mu_1}{1-\epsilon}\leq\frac{\mu_1}{1-\frac{1}{2}}=2\mu_1\leq 1$, as $0<\epsilon\leq\frac{1}{2}$ and $0<\mu_1\leq\frac{1}{2}$.} Then $u^*\geq 1-\epsilon$, as upon knowing $A$, Sender could have chosen a signaling scheme implying $p\parentheses*{s}=0$ w.p.~$\epsilon$ and $p\parentheses*{s}=\frac{\mu_1}{1-\epsilon}$ w.p.~$1-\epsilon$ (note that this scheme is, indeed, Bayes-plausible). Since $u\parentheses*{\pi,u_r}=\Pr_{s\sim\pi} \brackets*{p\parentheses*{s}\in A}\leq \Pr_{s\sim\pi} \brackets*{p\parentheses*{s}\in \parhalf*{\mu_1,1}}=\frac{1}{2}-\epsilon$, the regret is at least $\parentheses*{1-\epsilon}-\parentheses*{\frac{1}{2}-\epsilon}=\frac{1}{2}$, as needed.
\end{proof}
\begin{remark}
\label{rem:arb}
One can show that $\reg_{\arb}\geq\frac{1}{2}$ for any $n$ and any $\mu$ using similar methodology to the methodology we applied for $n=2$.
\end{remark}
Indeed, for a given $\pi$, Receiver should pick a hyperplane $H$ passing through $\mu$ that does not pass through any atom in $\pi$ besides, probably, $\mu$. Let $\Delta_1$ be one of the closed parts to which $H$ divides $\Delta\parentheses*{\Omega}$. If $\Pr_{s\sim\pi} \brackets*{p\parentheses*{s}\in \Delta_1}\leq \frac{1}{2}$, then by setting $A=\Delta_1$ Adversary ensures that Sender's expected utility is at most $\frac{1}{2}$, while the no-information scheme could have achieved utility of $1$; and if $Pr_{s\sim\pi} \brackets*{p\parentheses*{s}\in \Delta_1}= \frac{1}{2}+\epsilon$ for some $\epsilon\in\parhalf*{0,\frac{1}{2}}$, then setting $A$ to be a closed subset of $\Delta\parentheses*{\Omega}\setminus\Delta_1$ specified by a close enough hyperplane $H'$ to $H$ ensures that while $u\parentheses*{\pi,u_r}\leq\Pr_{s\sim\pi} \brackets*{p\parentheses*{s}\notin \Delta_1}=\frac{1}{2}-\epsilon$, Sender could have achieved utility of $1-\epsilon$ had she known $A$ by an appropriate binary-signal scheme.

Another case in which we characterize the optimal persuasion explicitly is that of ternary state space and uniform prior.
\begin{proposition}
\label{pro:arbitrary-ternary}
For $n=3$ and $\mu=\parentheses*{\frac{1}{3},\frac{1}{3},\frac{1}{3}}$: $\reg_{\arb}=\frac{1}{2}$. 
\end{proposition}
Although Proposition~\ref{pro:arbitrary-ternary} deals with a specific case, its proof is not trivial at all. The explicit characterizations of the regret in case of ternary state space with non-uniform prior and in case of any larger number of states remain open problems. The formal proof of Proposition~\ref{pro:arbitrary-ternary} is relegated to Appendix~\ref{ap:arbitrary-ternary}. We briefly describe here the main proof ideas.
\begin{proof}[Proof idea of Proposition~\ref{pro:arbitrary-ternary}]
By Remark~\ref{rem:arb}, Adversary can ensure a regret of at least $\frac{1}{2}$. To achieve such a regret, Adversary partitions the simplex $\Delta\parentheses*{\brackets*{3}}$ into two regions by drawing a line through the prior $\mu$ (which is the centroid of the simplex). Thereafter, Adversary sets $A$ to be the half-space that has a weight not above $\frac{1}{2}$ according to $\pi$ (for simplicity, here we ignore the possibility that $\pi$ has an atom on $\mu$).

The above argument can be applied for \emph{every} line passing through $\mu$ which does not contain atoms of $\pi$. This hints at the desirable distributions of posteriors for a regret-minimizing $\pi$: it should assign an equal weight of $\frac{1}{2}$ for every half-space that contains $\mu$ on the boundary. Another useful intuition from Proposition~\ref{pro:arbitrary-binary} is that it is worthwhile for Sender to "push" the weight to the boundary. This brings us to the (somewhat educated) guess of considering the \emph{unique} distribution that is supported on the boundary of $\Delta\parentheses*{\brackets*{3}}$, $\braces*{\parentheses*{p_1,p_2,p_3}\in\Delta\parentheses*{\brackets*{3}}:p_1=0 \text{ or } p_2=0 \text{ or } p_3=0}$, which assigns probability weight of exactly $\frac{1}{2}$ to every half-space containing $\mu$ as a boundary point. The complicated part of the proof is showing that this signaling scheme indeed guarantees a regret of $\frac{1}{2}$ to Sender. Namely, to show that Adversary cannot get a regret greater than $\frac{1}{2}$ by choosing the adoption region to be a half-space that does not contain $\mu$. This follows from geometric arguments and analytical geometry.
\end{proof}
The fact that the regret does not increase from binary to ternary state space is somewhat misleading. The following theorem provides a negative asymptotic result for large state spaces and an arbitrary prior.
\begin{theorem}
\label{thm:arbitrary-asy}
For every number of states $n$ and any prior $\mu$: $1-\frac{2}{\sqrt{n}}\leq \reg_{\arb} \leq 1-\frac{1}{4n^2}$.
\end{theorem}
Theorem~\ref{thm:arbitrary-asy} indicates that not knowing Receiver's utility in large state spaces might be very costly for Sender. For every signaling scheme, there are Receiver's utilities for which adoption occurs with probability almost $0$, while knowing Receiver's utility allows Sender to get adoption probability of almost $1$. This negative result is not very surprising, as we focus here on the most general setting and assume nothing on Receiver's utility. Now we describe the proof idea; the full proof appears in Appendix~\ref{ap:arb}.
\begin{proof}[Proof idea of Theorem~\ref{thm:arbitrary-asy}]
We prove that $\reg_{\arb} \leq 1-\frac{1}{4n^2}$ by describing an explicit signaling scheme ensuring such a regret bound. The more interesting part of the theorem is $\reg_{\arb}\geq 1-\frac{2}{\sqrt{n}}$. To show this part, we consider a specific scenario in which Sender knows that most states ($\sim n-\sqrt{n}$) are \emph{normal} -- they yield a constant negative Receiver's utility upon adoption. Among the remaining $\sim\sqrt{n}$ states, a single \emph{good} state with very high utility is hidden, while the other states are \emph{bad}; for clarity of exposition, assume that for bad states Receiver's utility upon adoption is $-\infty$. Had Sender known which state is the good one, she could have pooled it together with all the normal states; it would have been incentive-compatible for Receiver to adopt after receiving the signal that the state is not bad.

However, Adversary selects the state types uniformly at random; in particular, our ignorant Sender does not know which state is good. The support of any posterior distribution either contains a bad state or does not contain the good state with a high probability, which would cause Receiver to reject.
\end{proof}
\subsection{Monotonic Utilities}
\label{sub:mon}
Let $\reg_{\mon}$ be the regret for monotonic utilities. That is, in Definition~\ref{def:reg_pi}, the supremum is over all Receiver's utilities that are non-decreasing in the state of nature. In this setting, we prove a positive result that provides a full explicit characterization of the regret, for every number of states $n$ and for every prior $\mu$. We saw in the previous subsection that $n$ plays a significant role for arbitrary utilities. In particular, as the number of states increases, so does the uncertainty of Sender in the arbitrary utility case. Hence, one might expect the regret to always increase with the number of states. Surprisingly, our result shows that this intuition is wrong: Sender's partial knowledge on the monotonicity of Receiver's utility turns out to be sufficient to have as good guarantees as for a binary state space.

Recall that $\mu_n$ is the prior probability of state $n$; for monotonic utilities, Receiver gets the (weakly) largest utility upon adoption in this state. The following theorem gives a full characterization of $\reg_{\mon}$ in terms of $\mu_n$. In particular, $\reg_{\mon}$ does not depend on $n$. The regret as a function of $\mu_n$ is demonstrated in Figure~\ref{fig:regret}.
\begin{theorem}
\label{thm:mon}
For every number of states $n$ and every prior $\mu=\parentheses*{\mu_1,...,\mu_n}$:
\begin{align*}
  \reg_{\mon}= \begin{cases}
  \frac{1}{e} &\text{ if } \mu_n \leq \frac{1}{e} \\
  -\mu_n\ln \mu_n &\text{ if } \mu_n > \frac{1}{e}.
  \end{cases}
\end{align*}
\end{theorem}
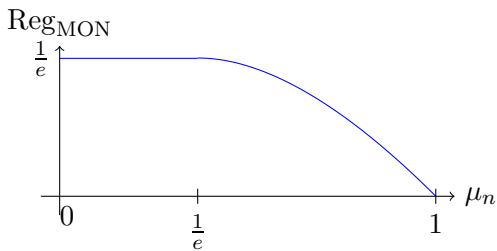
\begin{figure}[h]
\caption{The regret as a function of the prior probability of the highest state.}
    \label{fig:regret}
    \centering
    \begin{tikzpicture}[scale=5]
    \draw[->] (-0.05,0) -- (1.05,0);
    \draw[->] (0,-0.05) -- (0,0.4);
    \draw[blue] plot [domain=0.367:1,samples=100] (\x,{-\x*ln(\x)});
    \draw[blue] (0,0.367)--(0.367,0.367);

    \draw (0.367,-0.02)--(0.367,0.02);
    \draw (1,-0.02)--(1,0.02);
    
    \node[below] at (0.367,-0.02) {$\frac{1}{e}$};
    
    \node[below] at (1,-0.02) {1};
    
    \node[right] at (1.05,0) {$\mu_n$};
    
    \node[above] at (0,0.4) {$\reg_{\mon}$};
    
    \node[below] at (0.02,0) {0};
    
    \node[left] at (0,0.367) {$\frac{1}{e}$};
    
    \end{tikzpicture}
\end{figure}
An intuition for why parameters other than $\mu_n$ turn out to be irrelevant appears in the proof sketch below; the relevance of $\mu_n$ is intuitively connected to the observation that Adversary must choose $u_r\parentheses*{\omega_n,1}\geq 0$ to achieve a regret greater than 0 (as otherwise, Receiver would never adopt). Sender can utilize this observation in her favour: e.g., truthfully revealing that $\omega=n$ when it occurs yields her a utility of at least $\mu_n$. As we shall see, she can utilize it in a more clever manner, which decreases the regret even more. A possible interpretation of Theorem~\ref{thm:mon} might be that unless the highest state is very likely (its prior probability is greater than $\frac{1}{e}$), the regret equals $\frac{1}{e}$ regardless of the prior and of the number of states. Below we present the proof idea; for the full proof, see Appendix~\ref{ap:mon}. That appendix also contains \emph{an explicit description of the optimal Sender's and Adversary's strategies} (Lemma~\ref{lem:threshold-game}).
\begin{proof}[Proof idea of Theorem~\ref{thm:mon}]
The proof relies on the knapsack characterization of optimal signaling schemes in the standard persuasion model. The characterization of optimal policies by~\cite{RenaultSV17} states that every persuasion problem in the standard setting with binary-action Receiver has an optimal $x$-threshold scheme for some $x$ (see Definition~\ref{def:threshold}).

In our setting, Receiver's utility is unknown to Sender; hence, the optimal threshold $x$ is unknown. Nevertheless, we make an educated guess that threshold policies remain useful -- and in fact, regret-minimizing -- also in our setting. More concretely, we consider signaling schemes in which the threshold $y$ is drawn \emph{at random}, and thereafter the $y$-threshold scheme is implemented. After restricting Sender's schemes to the class of (mixed) threshold schemes, we also can view Adversary's choice of utility $u_r$ simply as a choice of an optimal threshold $x\in \brackets*{0,1}$ (see Fact~\ref{fact:optimal-threshold}). Namely, instead of choosing a utility for Receiver, Adversary chooses the threshold $x$ that is optimal in the standard persuasion model. This reduces the complicated original zero-sum game $G$ (see Subsection~\ref{sub:adversarial}), in which Sender chooses a signaling scheme and adversary chooses a distribution over possible functions $u_r$, to a much simpler zero-sum game $G_{\mu_n}$, with both players simply choosing thresholds $x,y\in \brackets*{0,1}$.

The utility in $G_{\mu_n}$ is given by: 
\begin{align*}
    &g\parentheses*{x,y}:=\parentheses*{1-x} - \parentheses*{1-y}\textbf{1}_{y\geq x},
\end{align*}
where $1-x$ is $u^*$ (i.e., the optimal utility in the standard persuasion model) and $\parentheses*{1-y}\textbf{1}_{y\geq x}$ is the utility of our ignorant Sender: if $y\geq x$, then Receiver adopts after observing the high signal and Sender gets utility of $1-y$; if $y<x$, then Receiver rejects after observing the high signal (recall that $y$ is Receiver's indifference point) and Sender gets a utility of $0$; anyway, Receiver rejects when observing the low signal.

In fact, not all thresholds $x\in \brackets*{0,1}$ might be optimal for some utility, but only $x\in \brackets*{0,1-\mu_n}$. Indeed, the highest-utility state should be included entirely in the knapsack, as otherwise adoption would never occur and the regret would be zero. As Sender's utility is decreasing in $y$ for $y\geq x$ and we know that the optimal threshold is in $\brackets*{0,1-\mu_n}$, we can restrict ourselves to considering Sender's and Adversary's thresholds $x,y$ from $\brackets*{0,1-\mu_n}$.

The analysis of $G_{\mu_n}$ (see Lemma~\ref{lem:threshold-game} in Appendix~\ref{ap:mon}) leads to the game value that appears in the theorem, and also provides the optimal strategies for Sender and Adversary.

After having the (allegedly) optimal strategies of both players in hand, we still have to verify that Sender cannot ensure a smaller regret than the value of $G_{\mu_n}$. To this end, we show that for a simple, yet optimal strategy of Adversary, Sender cannot gain more than the value of $G_{\mu_n}$ even by choosing arbitrary (not necessarily threshold) schemes.
\end{proof}
It is interesting to note that the regret-minimizing signaling scheme randomizes the threshold over a continuum-sized support $\brackets*{0,\min\braces*{1-\mu_n,1-\frac{1}{e}}}$, and, in particular, uses a continuum of signals; this stands in a sharp contrast to standard persuasion in which binary signals suffice to get the optimal utility. Inside the continuum-sized support, the density function is given by $f\parentheses*{y}:=\frac{1}{1-y}$ and, in particular, lower thresholds (closer to no-information) are chosen with lower probability.\footnote{To be precise, the threshold distribution also has an atom of weight $1+\ln\mu_n$ on $1-\mu_n$ (i.e., Sender reveals whether $\omega$ is the highest-utility state $n$ or not) if $\mu_n>\frac{1}{e}$.} Figure \ref{fig:density} demonstrates the density function of Sender's thresholds for $\mu_n=\frac{1}{4}$.
\begin{figure}[h]
\caption{The density function of Sender's threshold when $\mu_n=\frac{1}{4}$.}
\label{fig:density}
      \centering
    \begin{tikzpicture}[xscale=4]
    \draw[->] (-0.05,0) -- (1.05,0);
    \draw[->] (0,-0.05) -- (0,3.0);
    \draw[blue] plot [domain=0.0:0.63,samples=100] (\x,{1/(1-\x)});
    
    \draw[dashed] (0.63,0)--(0.63,2.72);

    \draw (0.12,-0.02)--(0.12,0.02);
    \draw (0.75,-0.02)--(0.75,0.02);
    \draw (1,-0.02)--(1,0.02);
    
    \node[below] at (0.63,-0.02) {$1-\frac{1}{e}$};
    
    \node[below] at (1,-0.02) {1};
    
    \node[right] at (1.05,0) {Sender's threshold $y$};
    
    \node[above] at (0,3.0) {$f$};
    
    \node[below] at (0.02,0) {0};
    
    \end{tikzpicture}
\end{figure}
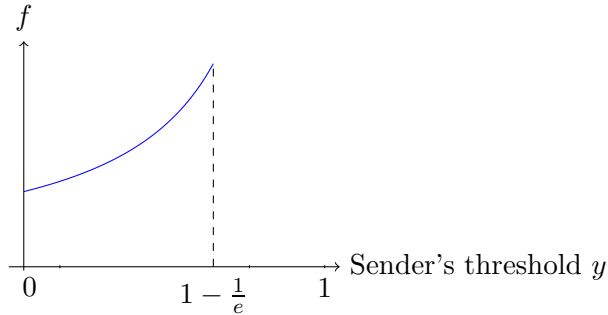

Note that our optimal robust signaling scheme is constructive and has a closed-form description, and is, therefore, polynomial-time computable. For example, for $\mu_n\leq 1/e$, one can implement the optimal robust scheme -- characterized in Lemma~\ref{lem:threshold-game} -- by sampling $z\sim U\brackets*{0,1}$ and adopting the $F^{-1}\parentheses*{z}$-threshold scheme, where $F\parentheses*{y}=-\ln\parentheses*{1-y}$ for $y\in\brackets*{0,1-1/e}$.
\subsection{Multidimensional Monotonic Utilities}
\label{sub:md}
In this setting, we denote the regret by $\reg_{\mul}$.
We assume that $\Omega=\times_{j\in \brackets*{k}} \Omega_j$, where $\Omega_j=\brackets*{n_j}$. We require Receiver's utility upon adoption to be monotonic in every dimension $j\in \brackets*{k}$. Namely, $\forall 1\leq j\leq k\;\;\forall\omega_j'\leq \omega_j''\;\;\forall\omega_{-j}': u_r\parentheses*{\omega_j'\omega_{-j}',1}\leq u_r\parentheses*{\omega''_j\omega_{-j}',1}$. We refer to the dimension $k$ as a constant. Indeed, for the interpretation of the dimension as the number of product attributes, it is natural to assume that $k$ is not too large. The parameters that determine the size of the problem are $n_1,...,n_k$ (i.e., the numbers of quality levels of each attribute/dimension). We show that the prior plays a significant role in determining whether a constant regret (i.e., bounded away from one) can be guaranteed: on the one hand, for arbitrary priors, the minimal regret might not be bounded away from $1$; on the other hand, for \emph{product} priors, a constant regret can be guaranteed.

For general priors, we have the following corollary from the arbitrary utility case.
\begin{corollary}
\label{cor:md-general-prior}
For $k=2$ and $n_1=n_2=m$ (for any $m\geq 1$), there exists a prior $\mu_m \in \Delta\parentheses*{\brackets*{m}^2}$ for which $Reg_{\mul}\parentheses*{m} = 1-O\parentheses*{\frac{1}{\sqrt{m}}}$.
\end{corollary}
Corollary~\ref{cor:md-general-prior} follows from Theorem~\ref{thm:arbitrary-asy} and the fact that one can assign arbitrarily small probability weights in the prior to all the states except for the $m$ states $\braces*{\parentheses*{i,j}\in \Delta\parentheses*{\brackets*{m}^2}:i+j=m+1}$. On these $m$ states, the monotonicity in each dimension imposes no restriction, while the remaining states have tiny weights, so their effect on the regret is negligible.

Therefore, to obtain positive results, one must restrict attention to particular classes of priors. One natural class is \emph{product priors}: $\mu=\times_{j\in \brackets*{k}} \mu_j $ for some $\mu_j \in \Delta(\brackets*{n_j})$. For this class of priors, we have the following result, which might be viewed as positive for small values of $k$.
\begin{proposition}
\label{pro:pos-md}
For every $k$, every sequence of attribute quality level amounts $n_1,...,n_k$ and every product prior $\mu =\times_{j\in \brackets*{k}} \mu_j$: $Reg_{\mul} \leq 1-2^{-k}$.
\end{proposition}
The full proof appears in Appendix~\ref{ap:mul}; here we only sketch the main ideas. We conjecture that a similar positive result can be obtained not only for product priors, but also for positively correlated priors.
\begin{proof}[Proof idea of Proposition~\ref{pro:pos-md}]
We present a simple signaling scheme that achieves the desired regret: Sender reveals whether the product quality in all the attributes is above median, where median is calculated according to the prior $\mu_j$. Depending on $u_r$, such a scheme either always leads to adoption if all the attribute qualities are above median -- yielding expected Sender's utility of at least $2^{-k}$ -- or not -- and then even had Sender known $u_r$, her expected utility could not have exceeded $1-2^{-k}$, as adoption could never occur with all the attribute qualities being below median.
\end{proof}
\section{Adversarial Approximation}
\label{sec:other-robustnesses}
In the adversarial approximation approach, the performance of a signaling scheme $\pi$ is measured by $\frac{u\parentheses*{\pi,u_r}}{u^*\parentheses*{u_r}}$, i.e., the \emph{ratio} between the utility Sender achieves with the scheme $\pi$ and the optimal utility Sender can get upon knowing $u_r$. If $u^*\parentheses*{u_r}=0$ (i.e., even the knowledgeable Sender cannot achieve any positive utility), we define the ratio to be $1$.

Analogically to Definitions~\ref{def:reg_pi} and~\ref{def:reg} of the additive regret, define the adversarial approximation guarantee via:
\begin{align*}
    &\apr\parentheses*{\pi}:=\inf_{u_r}\braces*{\frac{u\parentheses*{\pi,u_r}}{u^*\parentheses*{u_r}}} \text{ and } \\
    &\apr:=\sup_{\pi} \apr\parentheses*{\pi}.
\end{align*}
Similarly to the regret, we extend the definition of $\apr$ to $\apr_{\arb}$, $\apr_{\mon}$ and $\apr_{\mul}$ according to the class of considered Receiver's utility functions.

Note that $0\leq \apr \leq 1$, where $\apr \approx 0$ means that only a negligible fraction of the potential Sender's utility (i.e., with the knowledge of $u_r$) can be guaranteed. Our results in this section are largely negative, motivating our focus on regret minimization.

We start with the general Proposition~\ref{pro:reg-apr} that connects the notions of regret ($\reg$) and adversarial approximation guarantee ($\apr$).
\begin{proposition}
\label{pro:reg-apr}
$\apr \leq \frac{1}{\reg}-1$.
\end{proposition}
This fact holds in any adversarial setting with utilities (revenues) in $\brackets*{0,1}$ (and not only in our context of persuasion).
\begin{proof}[Proof of Proposition~\ref{pro:reg-apr}]
By Definitions~\ref{def:reg_pi} and~\ref{def:reg}, for every signaling scheme $\pi$ and every $\epsilon>0$, there exists $u_r^\epsilon$ s.t.~$u^*\parentheses*{u_r^\epsilon}-u\parentheses*{\pi,u_r^\epsilon}\geq \reg-\epsilon$. In particular, $u^*\parentheses*{u_r^\epsilon} \geq \reg-\epsilon$ and $u\parentheses*{\pi,u_r^\epsilon}\leq 1-\reg+\epsilon$. Furthermore, 
$\apr\parentheses*{\pi}\leq\frac{u\parentheses*{\pi,u_r^\epsilon}}{u^*\parentheses*{u_r^\epsilon}}\leq\frac{1-\reg+\epsilon}{\reg-\epsilon}$. Taking $\epsilon\to 0$ yields that $\apr\parentheses*{\pi}\leq \frac{1-\reg}{\reg}=\frac{1}{\reg}-1$. As it holds for every $\pi$, the proposition follows.
\end{proof}
Proposition~\ref{pro:reg-apr} allows to translate all the negative results on the regret into negative results on adversarial approximation: if $\reg \to 1$, then $\apr \to 0$, with at least the same convergence rate. Concretely, the negative results for arbitrary utilities (Theorem~\ref{thm:arbitrary-asy}) and multidimensional monotonic utilities (Corollary~\ref{cor:md-general-prior}) have analogues in the adversarial approximation approach.

A natural question is whether our main positive result (Theorem~\ref{thm:mon}) remains valid for the adversarial approximation setting. Namely, can a constant fraction of the optimal Sender's utility be guaranteed without knowing the cardinal Receiver's preferences, but only the ordinal ones (i.e., monotonicity)? Our next result shows that the answer is negative. Furthermore, we provide an exact characterization for the value in this case, for every number of states $n$ and every prior $\mu$.
\begin{theorem}
\label{thm:apr}
For every number of states $n$ and every prior $\mu=\parentheses*{\mu_1,...,\mu_n}$: $\apr_{\mon}=\frac{1}{1+\ln\frac{1}{\mu_n}}$.
\end{theorem}
This result indicates that for large state spaces (under the natural assumption that $\mu_n \to 0$), a constant fraction of the optimal utility \emph{cannot} be guaranteed. Namely, unlike the regret approach in which Sender can ensure a regret bounded away from $1$, Sender cannot guarantee a constant approximation ratio. Now we describe the proof idea for Theorem~\ref{thm:apr}; for the full proof, see Appendix~\ref{ap:apr}.
\begin{proof}[Proof idea of Theorem~\ref{thm:apr}]
Similar arguments to those in Theorem~\ref{thm:mon} proof apply. Recall that the proof of Theorem~\ref{thm:mon} simplifies the complex zero-sum game $G$ with the action sets of mixtures over Receiver's utility for Adversary and signaling schemes for Sender. The key observation is that one can reduce $G$ to a much simpler zero-sum game in which both players choose thresholds (i.e., real numbers in $\brackets*{0,1-\mu_n}$).

The same idea works for Theorem~\ref{thm:apr}; the only difference is in the utility function in the simplified game, which should be $\frac{\parentheses*{1-y}\textbf{1}_{y\geq x}}{1-x}$. In this game, $x$ is the action of Adversary (the optimal threshold choice, as described in Fact~\ref{fact:optimal-threshold}), while $y$ is the threshold specifying Sender's strategy. Note that $u^*=1-x$, while Sender's utility in the $y$-threshold scheme is $\parentheses*{1-y}\textbf{1}_{y\geq x}$.%
\footnote{In this zero-sum game, the $y$-player is the maximizer, while the $x$-player is the minimizer.} 
The reduction proof is analogous to Theorem~\ref{thm:mon}; all that remains is to analyze the continuum-action-set zero-sum game, which is done in Lemma~\ref{lem:threshold-game-apr} in Appendix~\ref{ap:apr}.
\end{proof}
\section{Conclusions and Future Work}
\label{sec:future}
In this paper, we adopt the adversarial regret minimization approach to Bayesian persuasion, and prove that while in the most general case the regret approaches $1$ as the number of states grows large, assuming that Sender knows Receiver's ordinal preferences upon adoption ensures that the regret is always at most $\frac{1}{e}$. We further provide an explicit formula for the regret, and describe the optimal Sender's and Adversary's strategies. We also study multidimensional monotonic utilities and upper bound the regret in this setting for product priors, while showing that for general priors the regret might tend to $1$.

Some problems remain open for the setting that we study. In particular, exact characterization of the value (and the regret-minimizing schemes) in the arbitrary utility case for more than two states (the only exception for which we succeeded to perform such an analysis is the case of uniform prior with ternary state space). The precise asymptotic convergence rate of $\reg_{\arb}$ to $1$ is also unknown to us -- we only bound it between $\theta\parentheses*{\frac{1}{\sqrt{n}}}$ and $\theta\parentheses*{\frac{1}{n^2}}$.

There are also quite a few natural extensions of the model: considering Receiver with more than two actions and considering a setting with multiple Receivers.

Finally, the minimax robust approach has not been discussed in this paper. Recall that in the minimax approach, Sender aims to maximize her expected utility over the signaling scheme and the unknown Adversary's strategy -- that is, there is no benchmark for the performance of the signaling scheme. Unfortunately, in the absence of any Sender's knowledge about Receiver's utility, the minimax approach is meaningless: Adversary can choose $u_r$ s.t.~Receiver never adopts, yielding the minimax value of $0$. An interesting question we leave for future research is to find a well-motivated Sender's \emph{partial knowledge} about $u_r$, which would make the minimax approach meaningful. Similarly, it is interesting to understand whether extra assumptions on Receiver's utility, rather than just monotonicity, may yield a constant adversarial approximation ratio.
\bibliographystyle{plainnat}
\bibliography{persuasion-bib}
\appendix
\section{Proof of Proposition~\ref{pro:arbitrary-ternary}}
\label{ap:arbitrary-ternary}
\begin{proof}[Proof of Proposition~\ref{pro:arbitrary-ternary}]
By Remark~\ref{rem:arb}, it is enough to show that $\reg_{\arb}\parentheses*{\pi}\leq\frac{1}{2}$. Consider $\pi$ supported on the boundary of $\Delta\parentheses*{\Omega}$ that assigns to every line segment on the boundary probability proportional to the angle upon which this segment is seen from $\mu$. We shall prove, using analytical geometry, that for this $\pi$, $\reg_{\arb}\parentheses*{\pi}=\frac{1}{2}$. Let $\at:\parhalf*{-\infty,\infty}\to \brackhalf*{0,\pi}$ be defined as $\at\parentheses*{x}:=\arctan\parentheses*{x}+\pi$ for $x\in\parentheses*{-\infty,0}$, $\at\parentheses*{x}:=\arctan\parentheses*{x}$ for $x\in\brackhalf*{0,\infty}$ and $\at\parentheses*{\infty}:=\frac{\pi}{2}$.

Fix $A$, and let $l$ be the line (single-dimensional hyperplane) separating $A$ from $\Delta\parentheses*{\Omega}\setminus A$. Note that $\Delta\parentheses*{\Omega}$ is contained in a two-dimensional plane. Represent the elements of $\Delta\parentheses*{\Omega}$ by Cartesian coordinates (rather than the probabilities they assign to elements of $\Omega$). Parameterize its extreme points by $\omega_1=\parentheses*{0,0}$, $\omega_2=\parentheses*{1,0}$, $\omega_3=\parentheses*{\frac{1}{2},\frac{\sqrt{3}}{2}}$ .\footnote{We assume that the sides of the triangle $\Delta\parentheses*{\Omega}$ are of length $1$ when in fact they are of length $\sqrt{2}$; however, homothety (scratch) on $\Delta\parentheses*{\Omega}$ does not affect the regret.} Then $\mu=\parentheses*{\frac{1}{2},\frac{\sqrt{3}}{6}}$. If the separating line $l$ is disjoint to $\Delta\parentheses*{\Omega}$ -- the regret is $0$. Therefore, assume that it cuts at least two line segments from the boundary of $\Delta\parentheses*{\Omega}$ (possibly at the same point). Assume w.l.o.g.~that $l$ cuts $\omega_1\omega_2$ (the line connecting $\omega_1$ and $\omega_2$) and $\omega_1\omega_3$. Then $l$ contains the points $D:=\parentheses*{d,\sqrt{3}d}$ and $E:=\parentheses*{e,0}$ for some $0\leq d\leq \frac{1}{2}$ and $0\leq e\leq 1$. Assume that $\mu\notin A$ -- otherwise, the adoption probability assured by $\pi$ is at least $\frac{1}{2}$ and the regret is at most $\frac{1}{2}$.
\begin{figure}[H]
    \centering
    \caption{Parameterization of the points.}
    \label{fig:param}
    
    \begin{tikzpicture}[scale=4]
        \draw (0,0)--(1,0)--(0.5,0.86)--(0,0);
        \filldraw[gray] (0,0)--(0.6,0)--(0.2,0.344);
        \node at (0.2,0.15) {$A$};
        \node[below] at (0,0) {$\omega_1=(0,0)$};
        \node[below] at (1,0) {$\omega_2=(1,0)$};
        \node[above] at (0.5,0.86) {$\omega_3=(\frac{1}{2},\frac{\sqrt{3}}{2})$};
        \filldraw (0.6,0) circle(0.01);
        \node[below] at (0.6,0) {$E=(e,0)$}; 
        \filldraw (0.2,0.344) circle(0.01);
        \node[left] at (0.2,0.344) {$D=(d,\sqrt{3}d)$};
        \filldraw (0.5,0.29) circle(0.01);
        \node[above] at (0.5,0.29) {$\mu=(\frac{1}{2},\frac{\sqrt{3}}{6})$};
    \end{tikzpicture}
\end{figure}
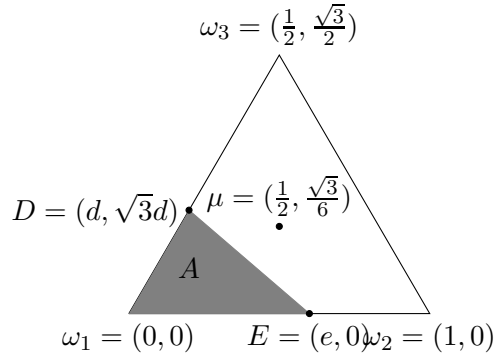
Consider now $\pi^*$, an optimal signaling scheme for a Sender who knows $A$. By pooling together all the probability mass in $A$ and all the probability mass in $\Delta\parentheses*{\Omega}\setminus A$ (separately), assume w.l.o.g.~that $\absolute*{\supp\parentheses*{\pi^*}}=2$.\footnote{$\absolute*{\supp\parentheses*{\pi^*}}>1$ since we assume that $\mu\notin A$.}

Let $s$ be the line segment connecting the two support elements of $\pi^*$. By Bayes-plausibility, it must pass through $\mu$. As $\pi^*$ is optimal, one can assume w.l.o.g.~that one support element is on $l$, while the other is on the boundary of $\Delta\parentheses*{\Omega}$ (outside $A$); furthermore, $\mu$ should divide $s$ in the maximal possible ratio. A straightforward (and a well-known) trigonometric exercise shows that the ratio in which a variable line segment $s$ passing through a fixed point, with endpoints on two fixed line segments $s_1,s_2$, is divided, is maximized when one of the endpoints of $s$ coincides with one of the endpoints of $s_1,s_2$. Therefore, one of the endpoints of $s$ is either an intersection point of $l$ with the boundary of $\Delta\parentheses*{\Omega}$, or a vertex of $\Delta\parentheses*{\Omega}$.

Assume first that exactly one extreme point of $\Delta\parentheses*{\Omega}$ is in $A$, and assume w.l.o.g.~that $\omega_1\in A$ (that is, $A$ is set by $l$ as in Figure~\ref{fig:param}). Qualitatively, we have $2$ different cases to check: either $E\in \supp\parentheses*{\pi^*}$ or $\omega_2\in \supp\parentheses*{\pi^*}$. Denote the second endpoint of $s$ (that is, the support element of $\supp\parentheses*{\pi^*}$ different from $E,\omega_2$) by $F$. Note that if $E\in \supp\parentheses*{\pi^*}$, then either $F\in\omega_2\omega_3$ or $F\in\omega_1\omega_3$.\footnote{$F$ cannot be on $\omega_1\omega_2$, as $s$ passes through $\mu$.}

Consider first the case $E=\parentheses*{e,0}\in \supp\parentheses*{\pi^*}$ and $F\in\omega_2\omega_3$ (see Figure~\ref{fig:case1}). That is, $F=\parentheses*{f,\sqrt{3}\parentheses*{1-f}}$, for $f=\frac{3-5e}{4-6e}$; note that as $F\in\omega_2\omega_3$ and $\mu\notin A$, we have $0\leq e\leq \frac{1}{2}$. The regret of $\pi$ for our fixed $A$ is:
\begin{align*}
    &\frac{\absolute*{\mu F}}{\absolute*{EF}}-\frac{\sphericalangle D\mu E}{2\pi}=\frac{1}{3-3e}-\frac{1}{2\pi}\at\parentheses*{\frac{\frac{\sqrt{3}/6}{1/2-e}-\frac{\sqrt{3}/6-\sqrt{3}d}{1/2-d}}{1+\frac{\sqrt{3}/6}{1/2-e}\cdot\frac{\sqrt{3}/6-\sqrt{3}d}{1/2-d}}}.
\end{align*}
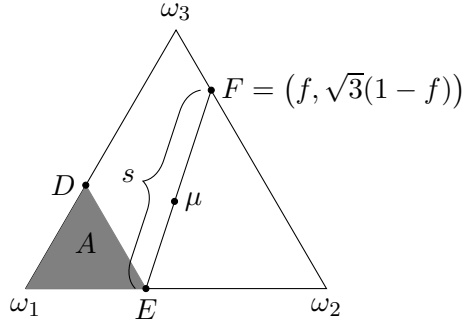
\begin{figure}[H]
    \centering
    \caption{An illustration of the case $E\in \supp\parentheses*{\pi^*}$, $F\in\omega_2\omega_3$.}
    \label{fig:case1}
    
    \begin{tikzpicture}[scale=4]
        \draw (0,0)--(1,0)--(0.5,0.86)--(0,0);
        \filldraw[gray] (0,0)--(0.4,0)--(0.2,0.344);
        \node at (0.2,0.15) {$A$};
        \node[below] at (0,0) {$\omega_1$};
        \node[below] at (1,0) {$\omega_2$};
        \node[above] at (0.5,0.86) {$\omega_3$};
        \filldraw (0.4,0) circle(0.01);
        \node[below] at (0.4,0) {$E$}; 
        \filldraw (0.2,0.344) circle(0.01);
        \node[left] at (0.2,0.344) {$D$};
        \filldraw (0.495,0.29) circle(0.01);
        \node[right] at (0.495,0.29) {$\mu$};
        
        \draw (0.4,0)--(0.617,0.66);
        \filldraw (0.617,0.66) circle(0.01);
        \node[right] at (0.617,0.66) {$F=\parentheses*{f,\sqrt{3}\parentheses*{1-f}}$};
        
        \draw [decorate,decoration={brace, amplitude=10pt},xshift=-1,yshift=0]
(0.4,0)--(0.617,0.66) node [black,midway,xshift=-15, yshift=5] 
{$s$};
    \end{tikzpicture}
\end{figure}
For a fixed $e$, the above expression is maximal for $d=0$.\footnote{One can obtain this step geometrically: setting $D=\omega_1$ does not affect the first term and does not increase the subtracted term.} Therefore (upon algebraic simplifications), it is enough to check that:
\begin{align*}
&\frac{1}{3-3e}-\frac{1}{2\pi}\at\parentheses*{\sqrt{3}\frac{e}{2-3e}}\leq\frac{1}{2}\iff\at\parentheses*{\sqrt{3}\frac{e}{2-3e}}-\frac{3e-1}{3-3e}\pi\geq 0.
\end{align*}
Note that the above expression is continuously differentiable for $e\in\brackets*{0,\frac{1}{2}}$, and its derivative is:\footnote{For $e\in\brackets*{0,\frac{1}{2}}$, the angle $\sphericalangle D\mu E=\sphericalangle \omega_1\mu E$ is acute; thus, $\at\parentheses*{\sqrt{3}\frac{e}{2-3e}}=\arctan\parentheses*{\sqrt{3}\frac{e}{2-3e}}$.}
\begin{align*}
    &\sqrt{3}\frac{1\cdot(2-3e)-e\cdot(-3)}{(2-3e)^2}\cdot\frac{1}{1+\parentheses*{\sqrt{3}\frac{e}{2-3e}}^2}-\frac{3(3-3e)-(3e-1)\cdot (-3)}{(3-3e)^2}\pi=\\
    &\frac{2\sqrt{3}}{3e^2+(2-3e)^2}-\frac{6\pi}{(3-3e)^2}< 0,
\end{align*}
where the last transition follows from: $2\sqrt{3}(3-3e)^2-6\pi\parentheses*{3e^2+(2-3e)^2}=(18\sqrt{3}-72\pi)e^2+e(72\pi-36\sqrt{3})+(18\sqrt{3}-24\pi)=(18\sqrt{3}-24\pi)(1-e)^2+24\pi e(1-2e)\leq_{e_{\opt}=\frac{2\pi-\sqrt{3}}{4\pi-\sqrt{3}}} \frac{4\pi^2(18\sqrt{3}-24\pi)}{(4\pi-\sqrt{3})^2}+24\pi\frac{2\pi\sqrt{3}-3}{(4\pi-\sqrt{3})^2}<0$.

Thus:
\begin{align*}
    &\at\parentheses*{\sqrt{3}\frac{e}{2-3e}}-\frac{3e-1}{3-3e}\pi\geq_{e_{\opt}=\frac{1}{2}}\at\parentheses*{\sqrt{3}\frac{0.5}{2-1.5}}-\frac{1.5-1}{3-1.5}\pi=\at\parentheses*{\sqrt{3}}-\frac{\pi}{3}=0,
\end{align*}
as needed.

Now, let us still assume that $E=\parentheses*{e,0}\in\supp\parentheses*{\pi^*}$, but $F\in\omega_1\omega_3$ (see Figure~\ref{fig:case2}). That is, $F=\parentheses*{f,\sqrt{3}\cdot f}$ for $f=\frac{e}{6e-2}$. Note that $\frac{1}{2}\leq e\leq 1$ (otherwise, $F$ is outside the line segment $\omega_1\omega_3$, a contradiction). Just as in the previous case, we obtain that for a fixed $e$, the regret is maximized when $d=0$. The regret equals for $d=0$ (here we assume that $e\neq 1$; for $e=1\implies E=\omega_2$, the adoption probability according to $\pi^*$ is $\frac{1}{3}$; thus, the regret is smaller than $\frac{1}{2}$):
\begin{align*}
    &\frac{\absolute*{\mu F}}{\absolute*{EF}}-\frac{\sphericalangle D\mu E}{2\pi}=\frac{1}{3e}-\frac{1}{2\pi}\at\parentheses*{\frac{\frac{\sqrt{3}/6}{1/2-e}-\frac{\sqrt{3}/6-\sqrt{3}d}{1/2-d}}{1+\frac{\sqrt{3}/6}{1/2-e}\cdot\frac{\sqrt{3}/6-\sqrt{3}d}{1/2-d}}}=_{d=0}\frac{1}{3e}-\frac{1}{2\pi}\at\parentheses*{\sqrt{3}\frac{e}{2-3e}}.
\end{align*}
To prove that this expression is at most $\frac{1}{2}$, it is enough to show that:
\begin{align*}
    &\at\parentheses*{\sqrt{3}\frac{e}{2-3e}}-\frac{2-3e}{3e}\pi\geq 0.
\end{align*}
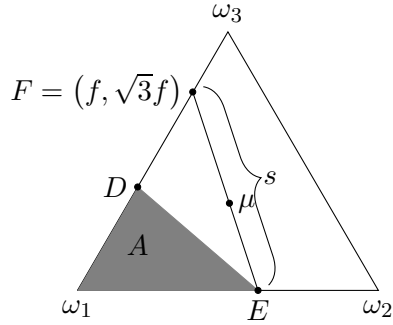
\begin{figure}[H]
    \centering
    \caption{An illustration of the case $E\in\supp\parentheses*{\pi^*}$, $F\in\omega_1\omega_3$.}
    \label{fig:case2}
    \begin{tikzpicture}[scale=4]
        \draw (0,0)--(1,0)--(0.5,0.86)--(0,0);
        \filldraw[gray] (0,0)--(0.6,0)--(0.2,0.344);
        \node at (0.2,0.15) {$A$};
        \node[below] at (0,0) {$\omega_1$};
        \node[below] at (1,0) {$\omega_2$};
        \node[above] at (0.5,0.86) {$\omega_3$};
        \filldraw (0.6,0) circle(0.01);
        \node[below] at (0.6,0) {$E$}; 
        \filldraw (0.2,0.344) circle(0.01);
        \node[left] at (0.2,0.344) {$D$};
        \filldraw (0.505,0.29) circle(0.01);
        \node[right] at (0.5,0.29) {$\mu$};
        
        \draw (0.6,0)--(0.383,0.66);
        \filldraw (0.383,0.66) circle(0.01);
        \node[left] at (0.383,0.66) {$F=\parentheses*{f,\sqrt{3}f}$};
        
        \draw [decorate,decoration={brace,mirror, amplitude=10pt},xshift=1,yshift=0.5]
(0.6,0)--(0.383,0.66) node [black,midway,xshift=13, yshift=3] 
{$s$};
    \end{tikzpicture}
\end{figure}
The last expression is increasing in $e$. Therefore, the minimum is attained for $e=\frac{1}{2}$, for which the condition holds:
\begin{align*}
    &\at\parentheses*{\sqrt{3}\frac{0.5}{2-1.5}}-\frac{2-1.5}{1.5}\pi=\at\parentheses*{\sqrt{3}}-\frac{\pi}{3}= 0.
\end{align*}
Now we should consider the case in which $E\notin \supp\parentheses*{\pi^*}$ (see Figure~\ref{fig:case3}). Then $\omega_2=\parentheses*{1,0}\in\supp\parentheses*{\pi^*}$ and $E\neq\omega_2$ (i.e., $e\neq 1$). As $E\in l$, $\omega_1\in A$ and $l$ separates $A$ from $\Delta\parentheses*{\Omega}\setminus A$, we must have $\omega_2\notin A$. Since the second endpoint of $s$, $F$, must belong to $A$ -- it must be on $l$. That is, $F=\parentheses*{f,\frac{1-f}{\sqrt{3}}}$ for $f=\frac{e-d-3de}{e-4d}$. Note that one must have $\frac{1}{4}\leq d\leq\frac{1}{2}$ (otherwise, $F$ cannot be on the segment $DE$) and $0\leq e\leq \frac{2d}{6d-1}$ (as $f\leq\frac{1}{2}$, since $\mu\notin A$). The regret is:
\begin{align*}
    &\frac{\absolute*{\mu\omega_2}}{\absolute*{F\omega_2}}-\frac{\sphericalangle D\mu E}{2\pi}.
\end{align*}
\begin{figure}[H]
    \centering
    \caption{An illustration of the case $E\notin \supp\parentheses*{\pi^*}$.}
    \label{fig:case3}
    \begin{tikzpicture}[scale=4]
        \draw (0,0)--(1,0)--(0.83,0.3);
        \draw (0.71,0.5)--(0.5,0.86)--(0,0);
        \filldraw[gray] (0,0)--(0.5,0)--(0.35,0.6);
        \node at (0.25,0.2) {$A$};
        \node[below] at (0,0) {$\omega_1$};
        \node[below] at (1,0) {$\omega_2$};
        \node[above] at (0.5,0.86) {$\omega_3$};
        \filldraw (0.5,0) circle(0.01);
        \node[below] at (0.5,0) {$E$}; 
        \filldraw (0.35,0.6) circle(0.01);
        \node[left] at (0.35,0.6) {$D$};
        \filldraw (0.5,0.29) circle(0.01);
        \node[right] at (0.5,0.29) {$\mu$};
        
        \draw (1,0)--(0.415,0.343);
        \filldraw (0.415,0.343) circle(0.01);
        \node at (0.65,0.4) {$F=(f,\frac{1-f}{\sqrt{3}})$};
        
        \draw [decorate,decoration={brace, amplitude=10pt},xshift=0,yshift=0]
(1,0)--(0.415,0.343) node [black,midway,xshift=-8, yshift=-10] 
{$s$};
    \end{tikzpicture}
\end{figure}
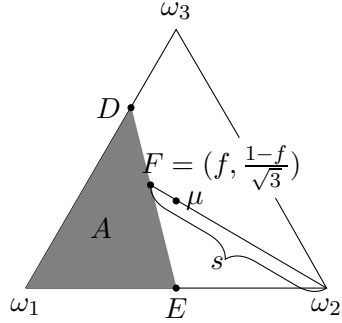
Fix $F$ and let $D,E$ vary so that $DE$ passes through $F$, while keeping $\frac{1}{4}\leq d\leq\frac{1}{2}$ and $0\leq e\leq \frac{2d}{6d-1}$. Consider the regret as a function of $d$. Note that $e=\frac{d\parentheses*{4f-1}}{f+3d-1}$. The regret equals:

\begin{align*}
    &\frac{\absolute*{\mu\omega_2}}{\absolute*{F\omega_2}}-\frac{\sphericalangle D\mu E}{2\pi}=\frac{1}{2\parentheses*{1-f}}-\frac{1}{2\pi}\at\parentheses*{\frac{\frac{\sqrt{3}/6}{1/2-e}-\frac{\sqrt{3}/6-\sqrt{3}d}{1/2-d}}{1+\frac{\sqrt{3}/6}{1/2-e}\cdot\frac{\sqrt{3}/6-\sqrt{3}d}{1/2-d}}}=\\
    &\frac{1}{2\parentheses*{1-f}}-\frac{1}{2\pi}\at\parentheses*{\sqrt{3}\cdot\frac{d(e-1/3)-e/6}{d(1-e)+e/2-1/3}}=\\
    &\frac{1}{2\parentheses*{1-f}}-\frac{1}{2\pi}\at\brackets*{\sqrt{3}\cdot\frac{d\parentheses*{d(4f-1)-\frac{f+3d-1}{3}}-\frac{d(4f-1)}{6}}{d\parentheses*{f+3d-1-d(4f-1)}+\frac{d(4f-1)}{2}-\frac{f+3d-1}{3}}}=\\
    &\frac{1}{2\parentheses*{1-f}}-\frac{1}{2\pi}\at\brackets*{2\sqrt{3}\parentheses*{2f-1}\frac{d^2-\frac{d}{4}}{d^2\parentheses*{4(1-f)}+d\parentheses*{3f-\frac{5}{2}}+\frac{1-f}{3}}}.
\end{align*}
The regret is continuous as a function of $d$; its derivative w.r.t.~$d$ has the same sign as (apart from finitely many values of $d$ for which the derivative of the regret is not defined):
\begin{align*}
    &-\parentheses*{2d-\frac{1}{4}}\brackets*{d^2\parentheses*{4(1-f)}+d\parentheses*{3f-\frac{5}{2}}+\frac{1-f}{3}}+\parentheses*{d^2-\frac{d}{4}}\brackets*{8d\parentheses*{1-f}+\parentheses*{3f-\frac{5}{2}}}=\\
    &d^2\parentheses*{\frac{3}{2}-2f}-d\parentheses*{\frac{2}{3}\parentheses*{1-f}}+\frac{1-f}{12}.
\end{align*}
Recall that $\frac{1}{4}\leq d\leq\frac{1}{2}$ and $\frac{1}{4}\leq f\leq\frac{1}{2}$ (because $\omega_2 F$ contains $\mu$ and $F$ is inside $\Delta\parentheses*{\Omega}$). The above function is decreasing for $\frac{1}{4}\leq d\leq\frac{2\parentheses*{1-f}}{3\parentheses*{3-4f}}$ and is increasing for $\frac{2\parentheses*{1-f}}{3\parentheses*{3-4f}}\leq d\leq \frac{1}{2}$.\footnote{Note that $\frac{1}{4}\leq \frac{2\parentheses*{1-f}}{3\parentheses*{3-4f}}\leq\frac{1}{2}$ for $\frac{1}{4}\leq f\leq\frac{1}{2}$.} Furthermore, its value at $d=\frac{1}{4}$ is:
\begin{align*}
    &\frac{1}{16}\parentheses*{\frac{3}{2}-2f}-\frac{1}{4}\cdot\frac{2}{3}\parentheses*{1-f}+\frac{1-f}{12}=\frac{1-4f}{96}\leq 0.
\end{align*}
Therefore, for a fixed $f$, the regret is maximized for a boundary value of $d$: either for $d=\frac{1}{4}$ or for $d=\frac{1}{2}$. If $d=\frac{1}{4}$, then necessarily $F=D$ and the regret is: $$\frac{\absolute*{\mu\omega_2}}{\absolute*{D\omega_2}}-\frac{\sphericalangle D\mu E}{2\pi}=\frac{2}{3}-\frac{\sphericalangle D\mu E}{2\pi}\leq\frac{2}{3}-\frac{\pi/3}{2\pi}=\frac{1}{2},$$
as needed. If $d=\frac{1}{2}$, then the value of the regret (after algebraic simplifications) is:
\begin{align*}
    &\frac{1}{2\parentheses*{1-f}}-\frac{1}{2\pi}\at\parentheses*{\frac{3\sqrt{3}\parentheses*{2f-1}}{2f+1}}
\end{align*}
The derivative of this expression is (note that it is well-defined for every $\frac{1}{4}\leq f\leq\frac{1}{2}$):
\begin{align*}
    &\frac{1}{2\parentheses*{1-f}^2}-\frac{3\sqrt{3}}{2\pi}\cdot\frac{1}{1+\parentheses*{\frac{3\sqrt{3}\parentheses*{2f-1}}{2f+1}}^2}\cdot\frac{2\parentheses*{2f+1}-2\parentheses*{2f-1}}{\parentheses*{2f+1}^2}=\\
    &\frac{1}{2\parentheses*{1-f}^2}-\frac{6\sqrt{3}}{\pi\parentheses*{\parentheses*{2f+1}^2+27\parentheses*{2f-1}^2}}=\frac{1}{2\parentheses*{1-f}^2}-\frac{3\sqrt{3}}{2\pi\parentheses*{28f^2-26f+7}}=\\
    &\frac{\parentheses*{28\pi-3\sqrt{3}}f^2-\parentheses*{26\pi-6\sqrt{3}}f+\parentheses*{7\pi-3\sqrt{3}}}{2\pi\parentheses*{1-f}^2\parentheses*{28f^2-26f+7}}>0,
\end{align*}
where the last transition holds as both the numerator and the denominator are positive for every $f\in\mathbb{R}$. Thus, given that $d=\frac{1}{2}$, the regret is maximal for $f=\frac{1}{2}$, for which the regret is exactly $\frac{1}{2}$, as desired.

To finish the proof, it remains to consider the case in which $A$ contains two different extreme points of $\Delta\parentheses*{\Omega}$.\footnote{If $A$ contains no extreme points or all the extreme points of $\Delta\parentheses*{\Omega}$ -- the regret is $0$.} If $s$ intersects $l$ at a boundary point of $\Delta\parentheses*{\Omega}$, then Adversary can increase the regret by rotating $l$ around that point in a way which decreases the area of $A$; this way $u^*$ cannot decrease, while $u\parentheses*{\pi,u_r}$ decreases. Therefore, it is enough to consider the case in which $s$ contains the unique extreme point of $\Delta\parentheses*{\Omega}$ that is not in $A$. Assume w.l.o.g.~that $\omega_2,\omega_3\in A$, while $\omega_1\notin A$ is an endpoint of $s$. Then the second element of $\supp\parentheses*{\pi}$, which belongs to $l$, is of the form $F=\parentheses*{f,\frac{f}{\sqrt{3}}}$ for $f=\frac{3de}{2d+e}$ (see Figure~\ref{fig:case4}). We must have $\frac{1}{2}\leq f\leq\frac{3}{4}$, as the segment $\omega_1 F$ contains $\mu$ and $F\in\Delta\parentheses*{\Omega}$. The regret is:
\begin{align*}
    &\frac{\absolute*{\mu \omega_1}}{\absolute*{F\omega_1}}-\frac{\sphericalangle D\mu E}{2\pi}=\frac{2d+e}{6de}-\frac{1}{2\pi}\at\parentheses*{\frac{\frac{\sqrt{3}/6}{1/2-e}-\frac{\sqrt{3}/6-\sqrt{3}d}{1/2-d}}{1+\frac{\sqrt{3}/6}{1/2-e}\cdot\frac{\sqrt{3}/6-\sqrt{3}d}{1/2-d}}}=\frac{2d+e}{6de}-\frac{1}{2\pi}\at\parentheses*{\sqrt{3}\cdot\frac{d(e-1/3)-e/6}{d(1-e)+e/2-1/3}}.
\end{align*}
\begin{figure}[H]
    \centering
    \caption{An illustration of the case $\omega_2,\omega_3\in A$.}
    \label{fig:case4}
    
    \begin{tikzpicture}[scale=4]
        \draw (0,0)--(1,0)--(0.83,0.3);
        \draw (0.71,0.5)--(0.5,0.86)--(0,0);
        \filldraw[gray] (0.8,0)--(0.42,0.72)--(0.5,0.86)--(1,0);
        \node at (0.6,0.55) {$A$};
        \node[below] at (0,0) {$\omega_1$};
        \node[below] at (1,0) {$\omega_2$};
        \node[above] at (0.5,0.87) {$\omega_3$};
        \filldraw (0.8,0) circle(0.01);
        \node[below] at (0.8,0) {$E$}; 
        \filldraw (0.42,0.72) circle(0.01);
        \node[left] at (0.42,0.72) {$D$};
        \filldraw (0.5,0.29) circle(0.01);
        \node[right] at (0.5,0.29) {$\mu$};
        
        \draw (0,0)--(0.62,0.36);
        \filldraw (0.62,0.36) circle(0.01);
        \node[right] at (0.65,0.4) {$F=\parentheses*{f,\frac{f}{\sqrt{3}}}$};
        
        \draw [decorate,decoration={brace, amplitude=10pt},xshift=0,yshift=0]
(0,0)--(0.63,0.36) node [black,midway,xshift=6, yshift=15] 
{$s$};
    \end{tikzpicture}
\end{figure}
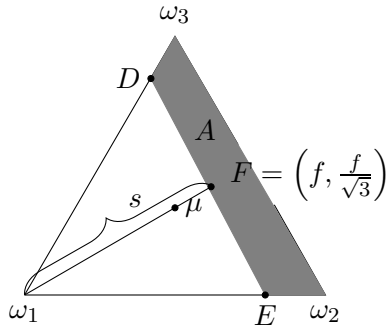
Set $x:=de$ and $y:=2d+e$. Then the above expression equals (note that $y=\frac{3x}{f}$):
\begin{align*}
    &\frac{y}{6x}-\frac{1}{2\pi}\at\parentheses*{\sqrt{3}\cdot\frac{x-y/6}{y/2-x-1/3}}=\frac{1}{2f}-\frac{1}{2\pi}\at\parentheses*{\sqrt{3}\cdot\frac{x-x/(2f)}{(3x)/(2f)-x-1/3}}=\\
    &\frac{1}{2f}-\frac{1}{2\pi}\at\parentheses*{\sqrt{3}\cdot\frac{x\parentheses*{2f-1}}{x\parentheses*{3-2f}-2f/3}}.
\end{align*}
Note that the angle $\sphericalangle D\mu E$ is obtuse for all feasible choices of $x,f$; thus, for a fixed $f$, the above expression is continuously differentiable w.r.t.~$x$. The derivative of the expression inside $\at\parentheses*{\cdot}$ w.r.t.~$x$ is:
\begin{align*}
    &\sqrt{3}\frac{\parentheses*{2f-1}\cdot\parentheses*{-2f/3}}{\parentheses*{x\parentheses*{3-2f}-2f/3}^2}<0,
\end{align*}
as $f>\frac{1}{2}$. Therefore, for a fixed $f$, the regret is maximal for the largest possible $x$. As $f=\frac{3x}{y}$, maximizing $x$ for a fixed $f$ is equivalent to maximizing $x$ for a fixed $y$. By the AM-GM inequality, $x\leq\frac{y^2}{8}$, with equality if and only if $2d=e$.\footnote{The geometric interpretation of this case is -- $DE$ is parallel to $\omega_3\omega_2$.} When $2d=e$, we have $f=\frac{3x}{y}=\frac{3y}{8}=\frac{3e}{4}$ and $x=\frac{e^2}{2}$; furthermore, we must have $\frac{1}{2}\leq e\leq 1$, as $\mu\notin A$. Assuming $2d=e$, the regret as a function of $e$ is:
\begin{align*}
    &\rho\parentheses*{e}:=\frac{2}{3e}-\frac{1}{2\pi}\at\parentheses*{\sqrt{3}\cdot\frac{e^2\parentheses*{3e/4-1/2}}{e^2\parentheses*{3/2-3e/4}-e/2}}=\frac{2}{3e}-\frac{1}{2\pi}\at\parentheses*{\sqrt{3}\cdot\frac{e\parentheses*{3e/4-1/2}}{e\parentheses*{3/2-3e/4}-1/2}}.
\end{align*}
We have:
\begin{align*}
    &\rho'\parentheses*{e}=-\frac{2}{3e^2}-\frac{\sqrt{3}}{2\pi}\frac{1}{1+3\parentheses*{\frac{e\parentheses*{3e/4-1/2}}{e\parentheses*{3/2-3e/4}-1/2}}^2}\cdot\\
    &\frac{\parentheses*{3e/2-1/2}\parentheses*{e\parentheses*{3/2-3e/4}-1/2}-e\parentheses*{3e/4-1/2}\parentheses*{3/2-3e/2}}{\parentheses*{e\parentheses*{3/2-3e/4}-1/2}^2}=\\
    &-\frac{2}{3e^2}-\frac{\sqrt{3}}{2\pi}\frac{3e^2-3e+1}{\parentheses*{3e^2/2-3e+1}^2+3\parentheses*{3e^2/2-e}^2}<0.
\end{align*}
Hence, the maximum of $\rho\parentheses*{e}$ on $\brackets*{\frac{1}{2},1}$ is obtained for $e=\frac{1}{2}$. Since $\rho\parentheses*{\frac{1}{2}}=\frac{1}{2}$, the regret is always at most $\frac{1}{2}$, which completes our proof.
\end{proof}
\section{Proof of Theorem~\ref{thm:arbitrary-asy}}
\label{ap:arb}
\begin{proof}[Proof of Theorem~\ref{thm:arbitrary-asy}]
We start by proving the lower bound. Assume $n\geq 16$ (otherwise the result follows from Remark~\ref{rem:arb}). Our proof constructs a difficult instance under which any algorithm suffers a regret of at least $1 - \frac{2}{\sqrt{n}}$. Fix $0<\delta<\frac{\min_{i\in\Omega} \mu_i}{2n}$. Set $\supp_{\delta}\parentheses*{p}:=\braces*{i\in\Omega:\Pr_{\omega'\sim p}\brackets*{\omega'=i}>\delta}$. Depending on what Receiver's utility a state leads to, we divide the $n$ states in our instance into three types: a single \emph{good} state, $n-\lfloor \sqrt{n} \rfloor-1$ \emph{normal} states  and $\lfloor \sqrt{n} \rfloor$ \emph{bad}  states. Sender's uncertainty about Receiver's utilities is captured by her uncertainty on the type of the true state of nature $\omega$.

Suppose that Adversary chooses a uniform permutation on $\Omega$; then she sets the first state in the permutation to be good, the next $n-\lfloor \sqrt{n} \rfloor-1$ states to be normal and the rest to be bad. Denote by $\mu_{\text{good}}$ the prior probability of the good state and by $\mu_{\text{normal}}$ -- the sum of the prior probabilities of the normal states. Suppose that for $i\in\Omega$, $u_r\parentheses*{i,1}$ equals $\frac{1}{\mu_{\text{good}}}$ if $i$ is good, $-\frac{1}{\mu_{\text{normal}}}$ if $i$ is normal and $-\frac{1}{\delta\cdot \mu_{\text{good}}}$ if $i$ is bad. The utility choices ensure that:
\begin{enumerate}
    \item The expected Receiver's utility knowing that the state is not bad is $0$.
    \item A posterior $p$ with a bad state in $\supp_{\delta}\parentheses*{p}$ leads to rejection.
    \item A posterior $p$ s.t.~the good state is not in $\supp_{\delta}\parentheses*{p}$ leads to rejection, as then $p$ must assign a probability of at least $\frac{1-\delta}{n-1}$ to a certain not good state, resulting in Receiver's utility being smaller than: \begin{align*}
        &\frac{\delta}{\mu_{\text{good}}}-\frac{1-\delta}{n-1}\cdot\frac{1}{\max\braces*{\delta\cdot\mu_{\text{good}}, \mu_{\text{normal}}}}=\frac{\delta}{\mu_{\text{good}}}-\frac{1-\delta}{\mu_{\text{normal}}\parentheses*{n-1}}<\\
 &\frac{1}{2n\mu_{\text{normal}}}-\frac{1}{\mu_{\text{normal}}\parentheses*{n-1}}+\frac{1}{2\mu_{\text{normal}}\parentheses*{n-1}}<0.
    \end{align*}
\end{enumerate}
Had Sender known the type of each state, she could have used a signaling scheme revealing whether the true state of nature $\omega$ is bad or not. This scheme would have made Receiver adopt exactly for not bad values of $\omega$. Thus, $u^*\geq 1-\frac{\lfloor\sqrt{n}\rfloor}{n}\geq 1-\frac{1}{\sqrt{n}}$.

We shall prove that our ignorant Sender cannot get a utility above $\frac{1}{\sqrt{n}}$; it would imply $\reg_{\arb}\geq \parentheses*{1-\frac{1}{\sqrt{n}}}-\frac{1}{\sqrt{n}}\geq 1-\frac{2}{\sqrt{n}}$, completing the proof. In fact, we shall show a stronger result: for every Receiver's posterior $p\in\Delta\parentheses*{\Omega}$, the probability of adoption, under Adversary's strategy, is at most~$\frac{1}{\sqrt{n}}$.

Indeed, fix some posterior $p\in\Delta\parentheses*{\Omega}$. We have the following cases:
\begin{itemize}
    \item If $\absolute*{\supp_{\delta}\parentheses*{p}}\leq \lfloor\sqrt{n}\rfloor$, 
    then the probability of having the good state in $\supp_{\delta}\parentheses*{p}$ is:\\$\frac{\absolute*{\supp_{\delta}\parentheses*{p}}}{n}\leq\frac{\sqrt{n}}{n}=\frac{1}{\sqrt{n}}$.
    \item Otherwise ($\absolute*{\supp_{\delta}\parentheses*{p}}> \lfloor\sqrt{n}\rfloor$), the probability of not having a bad state -- and meanwhile, having the good state -- in $\supp_{\delta}\parentheses*{p}$ is:
\begin{align*}
&\parentheses*{1-\frac{\lfloor\sqrt{n}\rfloor}{n}}\parentheses*{1-\frac{\lfloor\sqrt{n}\rfloor}{n-1}}...\parentheses*{1-\frac{\lfloor\sqrt{n}\rfloor}{n+1-\absolute*{\supp_{\delta}\parentheses*{p}}}}\cdot\min\parentheses*{\frac{\absolute*{\supp_{\delta}\parentheses*{p}}}{n-\lfloor\sqrt{n}\rfloor},1}\leq\\
&\parentheses*{1-\frac{\lfloor\sqrt{n}\rfloor}{n}}^{\absolute*{\supp_{\delta}\parentheses*{p}}}\cdot\frac{\absolute*{\supp_{\delta}\parentheses*{p}}}{n-\lfloor\sqrt{n}\rfloor}\leq\parentheses*{1-\frac{\lfloor\sqrt{n}\rfloor}{n}}^{\lfloor\sqrt{n}\rfloor+1}\cdot\frac{\lfloor\sqrt{n}\rfloor+1}{n-\lfloor\sqrt{n}\rfloor}\leq\\
&e^{-\frac{\lfloor\sqrt{n}\rfloor\parentheses*{\lfloor\sqrt{n}\rfloor+1}}{n}}\cdot\frac{\lfloor\sqrt{n}\rfloor+1}{n-\lfloor\sqrt{n}\rfloor}\leq\frac{1}{e^{1-1/\sqrt{n}}}\cdot\frac{\sqrt{n}+1}{n-\sqrt{n}} \leq_{n\geq 16}\frac{1}{\sqrt{n}},
\end{align*}
where the second inequality holds since $\parentheses*{1-\frac{\lfloor\sqrt{n}\rfloor}{n}}^x \cdot x$ is non-increasing as a function of the integer variable $x$ for $x\geq\lfloor\sqrt{n}\rfloor+1$, because $1-\frac{\lfloor\sqrt{n}\rfloor}{n}\leq\frac{x}{x+1}$ for $x\geq\lfloor\sqrt{n}\rfloor+1$.
\end{itemize}
Therefore, the expected adoption probability over Adversary's strategy at every posterior is at most $\frac{1}{\sqrt{n}}$, which completes the proof of the lower bound.

Let us prove now the upper bound. Note that the full revelation scheme \emph{fails} to provide any regret guarantee for priors that have tiny masses on some of the states. The regret guarantee of full revelation is as high as $1-\min_{i \in \brackets*{n}} \mu_i$.\footnote{Indeed, if for some $i\in\brackets*{n}$ we have $u_r\parentheses*{i,1}=\frac{1}{\mu_i}$ and $u_r\parentheses*{j,1}=-\frac{1}{1-\mu_i}$ for $j\neq i$, then the no-information scheme leads to adoption w.p.~1, while the full-revelation scheme leads to adoption w.p.~$\mu_i$; therefore, the regret is $1-\mu_i$.} We shall prove that a modification of the full-revelation scheme s.t.~every signal realization pools at most two states has a regret guarantee of $1-\frac{1}{4n^2}$ \emph{for every prior}.

Define $U:=\braces*{i\in\Omega:\mu_i\geq \frac{1}{2n}}$. Note that $\sum_{i\in U} \mu_i=1-\sum_{i\notin U} \mu_i\geq 1-\frac{n}{2n}=\frac{1}{2}$ (in particular, $U\neq\emptyset$).

Consider the signaling scheme $\pi$ with the set of signals $S=\braces*{s_{i}: i\in \Omega}\cup\braces*{s_{i,j}: i\in U, j\in\Omega\setminus U}$, s.t.:
\begin{itemize}
    \item For every $i\in U$, the signal $s_i$ is assigned a probability mass of $\parentheses*{1-\frac{1}{2n}}\mu_i$ out of the prior probability of $\mu_i$ for $\omega=i$.
    \item For every $i\in \Omega\setminus U$, the signal $s_i$ is assigned a probability mass of $\frac{\mu_i}{2}$ out of the prior probability for $\omega=i$.
    \item For every $i\in U$, $j\in \Omega\setminus U$, the signal $s_{i,j}$ is assigned a probability mass of $\frac{\mu_i}{2n\parentheses*{n-\absolute*{U}}}$ out of the prior probability for $\omega=i$ and a probability mass of $\frac{\mu_j}{2\absolute*{U}}$ out of the prior probability for $\omega=j$.
\end{itemize}
We claim that $\reg_{\arb}\parentheses*{\pi}\leq 1-\frac{1}{4n^2}$. Indeed, fix $u_r$.\footnote{The result for a mixture over possible functions $u_r$ would follow by taking the expectation over Adversary's strategy.} Define $T:= \braces*{i\in\Omega: u_r\parentheses*{i,1}\geq 0}$. Consider the following cases:
\begin{itemize}
    \item $T\cap U\neq\emptyset$. Since $s_i$ leads to adoption for every $i\in T$, in this case $u\parentheses*{\pi,u_r}\geq\frac{1}{2n}\cdot\parentheses*{1-\frac{1}{2n}}\geq \frac{1}{4n^2}$; thus, $\reg_{\arb}\parentheses*{\pi}=u^*\parentheses*{u_r}-u\parentheses*{\pi, u_r}\leq 1-\frac{1}{4n^2}$.
    \item $T\cap U=\emptyset$, and for every $i\in U$ there exists $j_i\in \Omega\setminus U$ s.t.~$\frac{\mu_i}{2n\parentheses*{n-\absolute*{U}}}\cdot u_r\parentheses*{i,1}+\frac{\mu_{j_i}}{2\absolute*{U}}\cdot u_r\parentheses*{j_i,1}\geq 0$. Then $s_{i,j_i}$ yields to adoption for every $i\in U$. We saw that $\sum_{i\in U} \mu_i\geq\frac{1}{2}$; thus, adoption occurs w.p.~at least $\sum_{i\in U} \frac{\mu_i}{2n\parentheses*{n-\absolute*{U}}} =\frac{\sum_{i\in U} \mu_i}{2n\parentheses*{n-\absolute*{U}}}\geq \frac{1}{4n^2}$, and again the regret is at most $1-\frac{1}{4n^2}$.
    \item $T\cap U=\emptyset$, and there exists $i\in U$ s.t.~for every $j\in \Omega\setminus U$: $\frac{\mu_i}{2n\parentheses*{n-\absolute*{U}}}\cdot u_r\parentheses*{i,1}+\frac{\mu_{j}}{2\absolute*{U}}\cdot u_r\parentheses*{j,1}<0$. Summing over $j\in \Omega\setminus U$ and multiplying by $2\absolute*{U}$ yields: $\frac{\absolute*{U}}{n}\cdot\mu_i\cdot u_r\parentheses*{i,1}+\sum_{j\in \Omega\setminus U}\mu_{j}\cdot u_r\parentheses*{j,1}<0$. If $U=\Omega$, then adoption never occurs for any signaling scheme and the regret is $0$. Otherwise, as $T\subseteq \Omega\setminus U$, we get that even upon the knowledge of $u_r$, a probability mass greater than $\parentheses*{1-\frac{\absolute*{U}}{n}}\cdot\mu_i\geq\frac{\mu_i}{n}\geq\frac{1}{2n^2}$ from the prior probability of $\mu_i$ for $\omega=i$ does not lead to adoption. Therefore, $u^*\leq 1-\frac{1}{2n^2}\leq 1-\frac{1}{4n^2}$, and $\reg_{\arb}\parentheses*{\pi}=u^*\parentheses*{u_r}-u\parentheses*{\pi, u_r}\leq 1-\frac{1}{4n^2}$.
\end{itemize}
In all the cases, $\reg_{\arb}\parentheses*{\pi}\leq 1-\frac{1}{4n^2}$, as desired.
\end{proof}

\section{Proof of Theorem~\ref{thm:mon}}
\label{ap:mon}
We start with the following lemma.
\begin{lemma}
\label{lem:threshold-game}
Let $\alpha\in \parentheses*{0,1}$ be a constant. Let $G_\alpha$ be a two-player zero-sum game with continuum action sets $X=Y=\brackets*{0,1-\alpha}$ of the $x$- (maximizing) and the $y$- (minimizing) players, respectively, and utility $$g\parentheses*{x,y}:=\parentheses*{1-x}-\parentheses*{1-y} \mathbf{1}_{y\geq x}.$$ Denote by $v=\val\parentheses*{G_\alpha}$ the value of $G_\alpha$.
\begin{itemize}
    \item For $\alpha\geq \frac{1}{e}$ we have $v=-\alpha \ln \alpha$. Furthermore, there exists an optimal strategy $o^*_x  \in \Delta\parentheses*{\brackets*{0,1-\alpha}}$ of the $x$-player that has an atom of weight $\alpha$ on $x=0$ and otherwise has the density function $f_X\parentheses*{x}:=\frac{\alpha}{\parentheses*{1-x}^2}$ over the entire segment $\brackets*{0,1-\alpha}$. Moreover, there exists an optimal strategy $ o^*_y  \in \Delta\parentheses*{\brackets*{0,1-\alpha}}$ of the $y$-player that has an atom of weight $1+\ln\alpha$ on $y=1-\alpha$ and otherwise has the density function $f_Y\parentheses*{y}:=\frac{1}{1-y}$ over the entire segment $\brackets*{0,1-\alpha}$.
    \item For $\alpha< \frac{1}{e}$ we have $v=\frac{1}{e}$. Furthermore, there exists an optimal strategy $ o^*_x  \in \Delta\parentheses*{\brackets*{0,1-\alpha}}$ of the $x$-player that has an atom of weight $\frac{1}{e}$ on $x=0$ and otherwise has the density function $f_X\parentheses*{x}:=\frac{1}{e\parentheses*{1-x}^2}$ over the segment $\brackets*{0,1-\frac{1}{e}}$. Moreover, there exists an optimal strategy $ o^*_y  \in \Delta\parentheses*{\brackets*{0,1-\alpha}}$ of the $y$-player that has the density function $f_Y\parentheses*{y}:=\frac{1}{1-y}$ over the segment $\brackets*{0,1-\frac{1}{e}}$.
\end{itemize}
\end{lemma}
\begin{proof}[Proof of Lemma~\ref{lem:threshold-game}]
Assume first that $\alpha\geq\frac{1}{e}$. We shall show that for the strategy couple $\parentheses*{o^*_x , o^*_y}$, both players are indifferent between all actions in $\brackets*{0,1-\alpha}$ and the payoff is $-\alpha\ln\alpha$; therefore, $\parentheses*{o^*_x , o^*_y}$ is an equilibrium and $v=-\alpha\ln\alpha$. Indeed, every fixed $x\in\brackets*{0,1-\alpha}$ yields expected (over $y$) payoff of:
\begin{align*}
    &\int_{0}^{1-\alpha} f_Y\parentheses*{y} g\parentheses*{x,y} dy + \Pr\brackets*{y=1-\alpha} g\parentheses*{x,1-\alpha} = \int_{0}^{1-\alpha} \frac{1}{1-y} \cdot \parentheses*{1-x} dy -\\
    &\int_{x}^{1-\alpha} \frac{1}{1-y} \cdot \parentheses*{1-y} dy +\parentheses*{1+\ln\alpha} \cdot \parentheses*{1-x-\alpha}=-\parentheses*{1-x}\ln\alpha-\parentheses*{1-\alpha-x}+\\
    &\parentheses*{1+\ln\alpha}\parentheses*{1-x-\alpha}=-\alpha\ln\alpha.
\end{align*}
Furthermore, every fixed $y \in\brackets*{0,1-\alpha}$ yields expected (over $x$) payoff of:
\begin{align*}
    &\int_{0}^{1-\alpha} f_X\parentheses*{x} g\parentheses*{x,y} dx + \Pr\brackets*{x=0} g\parentheses*{0,y} =
    \int_{0}^{1-\alpha} \frac{\alpha}{\parentheses*{1-x}^2} \cdot \parentheses*{1-x} dx -\\
    &\int_{0}^{y} \frac{\alpha}{\parentheses*{1-x}^2} \cdot \parentheses*{1-y} dx +\alpha \cdot y= -\alpha\ln\alpha-\parentheses*{1-y}\cdot\parentheses*{\frac{\alpha}{1-y}-\alpha}+\alpha \cdot y=-\alpha\ln\alpha,
\end{align*}
as desired.

Assume now that $\alpha<\frac{1}{e}$. We shall show that for the strategy couple $\parentheses*{o^*_x , o^*_y}$, both players are indifferent between the actions in $\brackets*{0,1-\frac{1}{e}}$ -- resulting in payoff of $\frac{1}{e}$ -- and worse off by taking other actions. Indeed, every fixed $x\in\brackets*{0,1-\frac{1}{e}}$ yields expected payoff of:
\begin{align*}
    &\int_{0}^{1-1/e} f_Y\parentheses*{y} g\parentheses*{x,y} dy = \int_{0}^{1-1/e} \frac{1}{1-y} \cdot \parentheses*{1-x} dy - \int_{x}^{1-1/e} \frac{1}{1-y} \cdot \parentheses*{1-y} dy =\\
    &\parentheses*{1-x}\cdot 1-\parentheses*{1-\frac{1}{e}-x}=\frac{1}{e},
\end{align*}
while for a fixed $x\in\parhalf*{1-\frac{1}{e},1-\alpha}$, the expected payoff is:
\begin{align*}
&\int_{1-1/e}^{1-\alpha} f_{Y}{dy}\parentheses*{y} g\parentheses*{x,y} d y= \int_{1-1/e}^{1-\alpha} \frac{1}{1-y} \cdot \parentheses*{1-x} dy =-\parentheses*{1+\ln\alpha}\parentheses*{1-x}<\\
&-\parentheses*{1+\ln\alpha}\parentheses*{1-\parentheses*{1-\frac{1}{e}}}<\frac{1}{e},
\end{align*}
as desired. Moreover, every fixed $y \in\brackets*{0,1-\frac{1}{e}}$ yields expected payoff of:
\begin{align*}
    &\int_{0}^{1-1/e} f_{X}\parentheses*{x} g\parentheses*{x,y} dx + \Pr\brackets*{x=0} g\parentheses*{0,y} =\int_{0}^{1-1/e} \frac{1}{e\parentheses*{1-x}^2} \cdot \parentheses*{1-x} dx -\\
    &\int_{0}^{y} \frac{1}{e\parentheses*{1-x}^2} \cdot \parentheses*{1-y} dx +\frac{1}{e} \cdot y=\frac{1}{e}\cdot 1-\frac{1-y}{e}\cdot\parentheses*{\frac{1}{1-y}-1}+\frac{1}{e} \cdot y=\frac{1}{e},
\end{align*}
while for a fixed $y\in\parhalf*{1-\frac{1}{e},1-\alpha}$ we have:
\begin{align*}
&\int_{0}^{1-1/e} f_{X}\parentheses*{x} g\parentheses*{x,y} dx + \Pr\brackets*{x=0} g\parentheses*{0,y} =\int_{0}^{1-1/e} \frac{1}{e\parentheses*{1-x}^2}\cdot\parentheses*{\parentheses*{1-x}-\parentheses*{1-y}} dx +\\
&\frac{1}{e} \cdot y=\frac{1}{e}\cdot\parentheses*{1-\parentheses*{1-y}\cdot\parentheses*{e-1}}+\frac{1}{e} \cdot y=\frac{2}{e}-1+y >\frac{1}{e},
\end{align*}
as desired.
\end{proof}
\begin{proof}[Proof of Theorem~\ref{thm:mon}]
Consider the two-player zero-sum game $G_{\mu_n}$ as described in Lemma~\ref{lem:threshold-game} and let $v$ be its value. We shall prove that $G$ (the zero-sum game interpretation of regret minimization; see Subsection~\ref{sub:adversarial}) also has value equal to $v$.

First, let us show that Sender can ensure regret of at most $v$ by using the signaling scheme $\pi_y$ defined as follows: Sender picks a random $y\sim o^*_y $ (where $ o^*_y $ is defined in Lemma~\ref{lem:threshold-game}); then she uses the $y$-threshold scheme (see Definition~\ref{def:threshold}).

Indeed, fix $u_r$. By Fact~\ref{fact:optimal-threshold}, it defines an optimal threshold $x=x\parentheses*{u_r}$ (in the standard persuasion model with the knowledge of Receiver's utility). Note that not all the values $x\in \brackets*{0,1}$ might be the optimal threshold for some $u_r$, but only $x\in \brackets*{0,1-\mu_n}$. Indeed, the highest-utility state $n$ should be included entirely in the knapsack, as otherwise adoption never occurs and the regret is $0$. Assume, therefore, that $x\in \brackets*{0,1-\mu_n}$.

We have $u^*=1-x$, and for every $y\in \brackets*{0,1-\mu_n}$: $u\parentheses*{\pi_y,u_r}=\textbf{1}_{y\geq x}$, as Receiver adopts if and only if $y\geq x$. Therefore, the regret for $\pi_y$ and $u_r$ is (where $g$ is as defined in Lemma~\ref{lem:threshold-game}):
\begin{align*}
    u^*\parentheses*{u_r}-u\parentheses*{\pi_y,u_r}=\mathbb{E}_{y\sim o^*_y }\brackets*{\parentheses*{1-x}-\parentheses*{1-y}\mathbf{1}_{y\geq x}}=
    \mathbb{E}_{y\sim o^*_y }\brackets*{g\parentheses*{x,y}}.
\end{align*}
By Lemma~\ref{lem:threshold-game}, the last expression is at most $v$. Therefore, Sender can ensure that the regret is at most $v$ for any fixed $u_r$; thus, it holds also for any mixture over $u_r$, as needed.

It remains to prove that Adversary can ensure a regret of at least $v$. Consider the following strategy of Adversary: she deterministically sets $u_r\parentheses*{i,1}:=-\mu_n$ for all $1\leq i\leq n-1$; then she chooses a random $t\sim o_x^*$ (with $o_x^*$ from Lemma~\ref{lem:threshold-game}) and sets $u_r\parentheses*{n,1}:=1-\mu_n-t$. For this Adversary's strategy, the optimal threshold $x$ is distributed according to $o_x^*$. To bound the regret that Sender can guarantee against this Adversary's strategy, fix a signaling scheme $\pi$. As $u_r\parentheses*{i,1}$ is the same for $1\leq i\leq n-1$, one can refer to the states $1,2,...,n-1$ as a single state that we call state $0$. Concretely, $\pi$ induces a signaling scheme $\pi'$ in a persuasion scenario with a binary-state space $\braces*{0,n}$ s.t.~$\mu_0:=1-\mu_n$; a posterior $p\parentheses*{s}\in \Delta\parentheses*{\brackets*{n}}$ is mapped to the posterior $\parentheses*{\sum_{1\leq i\leq n-1} p_i,p_n}$. Therefore, computing $\reg_{\mon}\parentheses*{\pi}$ in our original setting is reduced to computing $\reg_{\mon}\parentheses*{\pi}$ in this binary-state setting, with Adversary's strategy being setting $u_r\parentheses*{0,1}:=-\mu_n$ deterministically and choosing $u_r\parentheses*{n,1}=1-\mu_n-t$, where $t\sim o_x^*$.

From now on, we shall refer to a posterior $p\in \Delta\parentheses*{\braces*{0,n}}$ as a real number $q\in \brackets*{0,1}$, where $q:=p_0$. To understand the optimal Sender's utility in this binary-state persuasion problem, we compute her expected utility $u'\parentheses*{q}$ (when the expectation is over Adversary's mixed strategy) for each possible posterior $q\in \brackets*{0,1}$; then we evaluate the concavification of $u'$ at the prior $1-\mu_n$.

We shall prove that there exists Sender's best-reply signaling scheme that is a threshold signaling scheme. It would complete the proof, as the optimal threshold $x$ is distributed according to $o_x^*$, which gives for a Sender's $y$-threshold scheme expected regret of:
\begin{align*}
    \mathbb{E}_{x\sim o^*_x }\brackets*{\parentheses*{1-x}-\parentheses*{1-y}\mathbf{1}_{y\geq x}}=
    \mathbb{E}_{x\sim o^*_x}\brackets*{g\parentheses*{x,y}},
\end{align*}
(with $g$ defined in Lemma~\ref{lem:threshold-game}) which is, by Lemma~\ref{lem:threshold-game}, at least $v$ for every $y\in\brackets*{0,1-\mu_n}$.\footnote{We can assume w.l.o.g.~that $y\in\brackets*{0,1-\mu_n}$, as the optimal threshold $x$ is at most $1-\mu_n$; thus, choosing $y>1-\mu_n$ gives a greater regret than choosing $y=1-\mu_n$.}

Indeed, to understand Sender's best-reply  we shall consider a standard Bayesian persuasion instance in which Sender's utility is the expected (over Adversary's strategy) adoption probability. Sender's expected utility, as a function of Receiver's posterior $q$, is: $$u'\parentheses*{q}:=\Pr_{t\sim o^*_x }\brackets*{-\mu_n q+\parentheses*{1-\mu_n-t}\parentheses*{1-q}\geq 0}=\Pr_{t\sim o^*_x }\brackets*{t\leq 1-\frac{\mu_n}{1-q}}.$$

Straightforward calculations show that for $\mu_n\geq\frac{1}{e}$:
\begin{align*}
  u'\parentheses*{q}=\begin{cases}
  1-q &\text{ if } 0\leq q\leq\mu_0\\
  0 &\text{ if } \mu_0< q\leq 1,
  \end{cases}
 \end{align*}
while for $\mu_n<\frac{1}{e}$:
\begin{align*}
  u'\parentheses*{q}=\begin{cases}
  1 &\text{ if } 0\leq q<1-\mu_n e \\
  \frac{1-q}{\mu_n e} &\text{ if } 1-\mu_n e\leq q\leq\mu_0\\
  0 &\text{ if } \mu_0< q\leq 1.
  \end{cases}
 \end{align*}
\begin{figure}[h]
\caption{The function $u'\parentheses*{q}$ and its concavification. The function appears in blue; its concavification appears in red.}
    \label{fig:cav}
    \centering
    \begin{tikzpicture}[scale=3]
    \draw[->] (-0.05,0) -- (1.05,0);
    \draw[->] (0,-0.05) -- (0,1.05);
    \draw[blue] (0,1)--(0.5,0.5);
    \draw[blue] (0.5,0.01)--(1,0.01);
    \draw[red] (1.01,0)--(0.01,1);
    \draw (1,-0.02)--(1,0.02);
    \draw (0.5,-0.02)--(0.5,0.02);
    \node[below] at (0.5,-0.02) {$\mu_0$};
    \node[below] at (1,-0.02) {1};
    \node[right] at (1.05,0) {$q$};
    \node[above] at (0,1.05) {$u'$};
    \node[below] at (0.02,0) {0};
    \node[left] at (0,1) {$1$};
    \node[below] at (0.5,-0.2) {The case of $\mu_n\geq \frac{1}{e}$.};
    \end{tikzpicture}
    \begin{tikzpicture}[scale=3]
    \draw[->] (-0.05,0) -- (1.05,0);
    \draw[->] (0,-0.05) -- (0,1.05);
    \draw[blue] (0,1)--(0.46,1);
    \draw[blue] (0.46,1)--(0.8,0.37);

    \draw[blue] (0.8,0.01)--(1,0.01);
    
    \draw[red] (1,0.01)--(0.46,1.01);
    \draw[red] (0,1.01)--(0.46,1.01);
    \draw (1,-0.02)--(1,0.02);
    \draw (0.8,-0.02)--(0.8,0.02);
    
    \node[below] at (0.8,-0.02) {$\mu_0$};
    
    \node[below] at (1,-0.02) {1};
    
    \node[right] at (1.05,0) {$q$};
    
    \node[above] at (0,1.05) {$u'$};
    
    \node[below] at (0.02,0) {0};
    
    \node[left] at (0,1) {$1$};
    
    \node[below] at (0.5,-0.2) {The case of $\mu_n< \frac{1}{e}$.};
    
    \end{tikzpicture}
\end{figure}
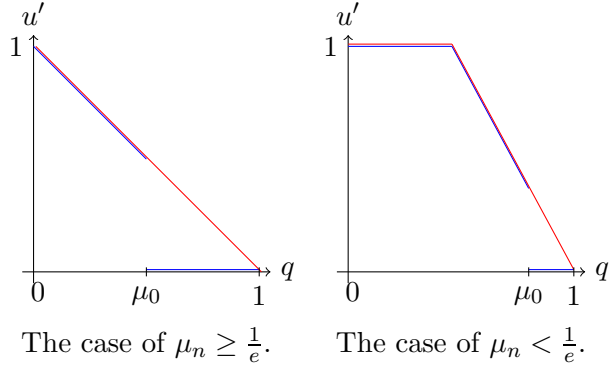
In both cases, $u'\parentheses*{q}=0$ for $q>\mu_0$ and $u'\parentheses*{q}>0$ otherwise. Therefore, the graph of the concavification of $u'$ (see Figure~\ref{fig:cav}) includes a line segment connecting the point $\parentheses*{q=1,u'\parentheses*{q}=0}$ with a point having a value of $q$ smaller than $\mu_0$. Thus, there exists an optimal Sender's signaling scheme with binary signals s.t.~one of them leads to the posterior $q=1$ (i.e., certainty that $\omega=0$); this scheme is a threshold scheme (see Definition~\ref{def:threshold}), as desired.
\end{proof}

\section{Proof of Proposition \ref{pro:pos-md}}
\label{ap:mul}
\begin{proof}[Proof of Proposition~\ref{pro:pos-md}]
To describe the proof, we use the \emph{fractional multidimensional knapsack approach}. Consider the $k$-dimensional cube $P:=\times_{j\in \brackets*{k}}\brackets*{1,n_j}$ with integer points representing the elements of $\Omega$. A \emph{knapsack} $K$ is a closed subset of $P$ s.t. for every $p\in K$, all the points that \emph{Pareto-dominate} $p$ are also in $K$.\footnote{A point $q=\parentheses*{q^1,...,q^k}$ \emph{Pareto-dominates} $p=\parentheses*{p^1,...,p^k}$ if $q^j\geq p^j$ for every $1\leq j\leq k$.} We define a \emph{knapsack signaling scheme} -- the multidimensional variant of a threshold scheme (see Definition~\ref{def:threshold}) -- as follows. Interpret $\Omega$, equipped with the prior $\mu$, as drawing uniformly $p\in P$ -- called the \emph{continuous state} -- s.t.~all the realizations in $\times_{j\in \brackets*{k}}\parhalf*{\sum_{l_j<m_j}\Pr_{\omega_j '\sim \mu_j}\brackets*{\omega_j '=l_j},\sum_{l_j\leq m_j}\Pr_{\omega_j '\sim \mu_j}\brackets*{\omega_j '=l_j}}$ correspond to the state $\parentheses*{m_1,...,m_k}$ (when $1\leq m_j\leq n_j$ for every $1\leq j\leq k$). For a knapsack $K$, the \emph{$K$-knapsack signaling scheme} is a binary-signal scheme revealing whether $p\in K$ or not.\footnote{The knapsack scheme specified by a knapsack with zero volume is the no-information signaling scheme.}

Consider the signaling scheme $\pi$ specified by the knapsack\\$K_{\pi}:=\braces*{\parentheses*{p^1,...,p^k}\in P: p^j\geq\frac{n_j+1}{2} \forall 1\leq j\leq k}$. We claim that $\pi$ ensures a regret of at most $1-2^{-k}$.

Indeed, if Adversary chooses $A$ s.t. every continuous state realization in $K_{\pi}$ leads to adoption -- Sender earns exactly $2^{-k}$. Since always $u^{*}\parentheses*{u_r}\leq 1$, the regret is at most $1-2^{-k}$, as desired.

Otherwise, there exists some continuous state realization $p=\parentheses*{p^1,...,p^k}\in K$ leading to rejection. As $\parentheses*{p^1,...,p^k}\in K$, we have $p^j\geq\frac{n_j+1}{2}$ for every $1\leq j\leq k$. Thus, no realization with the $j$-th coordinate being at most $\frac{n_j+1}{2}$ for every $j$ leads to adoption. Similarly to the single-dimensional setting, for a given $u_r$, there exists an optimal knapsack signaling scheme. Had Sender known $A$, her optimal knapsack strategy would have been specified by a knapsack $K^*$ disjoint to $S:=\braces*{\parentheses*{q^1,...,q^k}\in P: q^j\leq\frac{n_j+1}{2}\forall 1\leq j\leq k}$. Such a strategy yields adoption probability of at most $1-\frac{\vol\parentheses*{S}}{\vol\parentheses*{P}}=1-2^{-k}$. Therefore, the regret is at most $u^*\leq 1-2^{-k}$, as desired.
\end{proof}
\section{Proof of Theorem~\ref{thm:apr}}
\label{ap:apr}
We start with the following lemma.
\begin{lemma}
\label{lem:threshold-game-apr}
Let $\alpha\in \parentheses*{0,1}$ be a constant. Denote $\beta=\beta\parentheses*{\alpha}:=\frac{1}{1+\ln\frac{1}{\alpha}}$. Let $G_\alpha '$ be a two-player zero-sum game with continuum action sets $X=Y=\brackets*{0,1-\alpha}$ of the $x$- (minimizing) and the $y$- (maximizing) players, respectively, and utility $$h\parentheses*{x,y}:=\frac{\parentheses*{1-y} \mathbf{1}_{y\geq x}}{{1-x}}.$$ Then the value of $G_\alpha '$ is $\beta$. Furthermore, there exists an optimal strategy ${o^*_x}'  \in \Delta\parentheses*{\brackets*{0,1-\alpha}}$ of the $x$-player that has an atom of weight $\beta$ on $x=0$ and otherwise has the density function $f_X\parentheses*{x}:=\frac{\beta}{1-x}$ over the entire segment $\brackets*{0,1-\alpha}$. Moreover, there exists an optimal strategy ${o^*_y}'  \in \Delta\parentheses*{\brackets*{0,1-\alpha}}$ of the $y$-player that has an atom of weight $\beta$ on $y=1-\alpha$ and otherwise has the density function $f_Y\parentheses*{y}:=\frac{\beta}{1-y}$ over the entire segment $\brackets*{0,1-\alpha}$.
\end{lemma}
\begin{proof}[Proof of Lemma~\ref{lem:threshold-game-apr}]
We shall show that for the strategy couple $\parentheses*{{o^*_x}', {o^*_y}'}$, both players are indifferent between all actions in $\brackets*{0,1-\alpha}$ and the payoff is $\beta$; therefore, $\parentheses*{{o^*_x}' , {o^*_y}'}$ is an equilibrium and the game value is $\beta$. Indeed, every fixed $x\in\brackets*{0,1-\alpha}$ yields expected (over $y$) payoff of:
\begin{align*}
    &\int_{0}^{1-\alpha} f_Y\parentheses*{y} h\parentheses*{x,y} dy + \Pr\brackets*{y=1-\alpha} h\parentheses*{x,1-\alpha} = \int_{x}^{1-\alpha} \frac{\beta}{1-y}\cdot\frac{1-y}{1-x} dy +\beta\cdot \frac{1-\parentheses*{1-\alpha}}{1-x}=\\
    &\frac{\beta\parentheses*{1-\alpha-x}}{1-x}+\frac{\alpha\beta}{1-x}=\beta,
\end{align*}
as needed. Furthermore, every fixed $y \in\brackets*{0,1-\alpha}$ yields expected (over $x$) payoff of:
\begin{align*}
    &\int_{0}^{1-\alpha} f_X\parentheses*{x} h\parentheses*{x,y} dx + \Pr\brackets*{x=0} h\parentheses*{0,y} =
    \int_{0}^{y} \frac{\beta}{1-x} \cdot \frac{1-y}{1-x} dx +\beta\cdot\frac{1-y}{1-0}=\\
    &\beta\parentheses*{1-y}\parentheses*{\frac{1}{1-y}-1}+\beta\parentheses*{1-y}=\beta,
\end{align*}
as desired.
\end{proof}
\begin{proof}[Proof of Theorem~\ref{thm:apr}]
Using the notations of Lemma~\ref{lem:threshold-game-apr}, consider the two-player zero-sum game $G_{\mu_n}'$; by the lemma, it has value of $\beta:=\beta\parentheses*{\mu_n}=\frac{1}{1+\ln\frac{1}{\mu_n}}$. We claim that the two-player zero-sum game interpretation of adversarial approximation maximization, in which Sender is the maximizing player and her possible strategies are the signaling schemes, and Adversary is the minimizing player and her possible strategies are the mixtures over Receiver's utility functions, has value of $\beta$.

Using similar arguments to Theorem~\ref{thm:mon} proof, Sender can ensure adversarial approximation of at least $\beta$ by using the signaling scheme $\pi_y$ defined as follows: Sender picks a random $y\sim {o^*_y}' $; then she uses the $y$-threshold scheme. The only difference from Theorem~\ref{thm:mon} proof is that in the current proof, for a fixed $u_r$ and the corresponding optimal threshold $x$, the expected adversarial approximation over $\pi_y$ is $\mathbb{E}_{y\sim o^*_y }\brackets*{h\parentheses*{x,y}}$ (for $h$ from Lemma~\ref{lem:threshold-game-apr}).

It remains to prove that Adversary can ensure an adversarial approximation of at most $\beta$. As in Theorem~\ref{thm:mon} proof, it is enough to prove that Adversary can ensure an adversarial approximation of at most $\beta$ in a persuasion scenario with a binary-state space $\braces*{0,n}$ s.t.~$\mu_0:=1-\mu_n$ by using the following strategy: $u_r\parentheses*{0,1}:=-\mu_n$ deterministically; $u_r\parentheses*{n,1}:=1-\mu_n-t$, where $t\sim o_x^*$.

We shall refer to a posterior $p\in \Delta\parentheses*{\braces*{0,n}}$ as a real number $q\in \brackets*{0,1}$, where $q:=p_0$. As in Theorem~\ref{thm:mon} proof, we compute Sender's expected utility $u'\parentheses*{q}$ (the expectation is over Adversary's mixed strategy) for each possible posterior $q\in \brackets*{0,1}$, and then we evaluate the concavification of $u'$ at the prior $1-\mu_n$. Again, it is enough to prove that there exists Sender's best-reply signaling scheme that is a threshold scheme.

Indeed, to understand Sender's best-reply we can consider a standard Bayesian persuasion instance in which Sender's utility is the expected -- over Adversary's strategy -- ratio of the indicator whether adoption occurs at the posterior $q$ to the optimal adoption probability upon knowing $t$. Sender's expected utility, as a function of the posterior $q$, is: $$u'\parentheses*{q}:=\mathbb{E}_{t\sim o^*_x }\brackets*{\frac{\mathbbm{1}_{-\mu_n q+\parentheses*{1-\mu_n-t}\parentheses*{1-q}\geq 0}}{1-t}}=\mathbb{E}_{t\sim o^*_x }\brackets*{\frac{\mathbbm{1}_{t\leq 1-\frac{\mu_n}{1-q}}}{1-t}}.$$

Straightforward calculations show that (see Figure~\ref{fig:cav2}):
\begin{align*}
  u'\parentheses*{q}=\begin{cases}
  \frac{\beta\parentheses*{1-q}}{\mu_n} &\text{ if } 0\leq q\leq\mu_0\\
  0 &\text{ if } \mu_0< q\leq 1.
  \end{cases}
 \end{align*}
Therefore, there exists an optimal signaling scheme at the prior that uses the posterior $q=1$, which is, in particular, a threshold scheme, as needed.
\begin{figure}[H]
\caption{The function $u'$ (appears in blue) and its concavification (appears in red).}
    \label{fig:cav2}
    \centering
    \begin{tikzpicture}[scale=3]
    \draw[->] (-0.05,0) -- (1.05,0);
    \draw[->] (0,-0.05) -- (0,1.25);
    \draw[blue] (0,1.2)--(0.5,0.6);

    \draw[blue] (0.5,0.01)--(1,0.01);
    
    \draw[red] (1.01,0)--(0.01,1.2);
    \draw (1,-0.02)--(1,0.02);
    \draw (0.5,-0.02)--(0.5,0.02);
    
    \node[below] at (0.5,-0.02) {$\mu_0$};
    
    \node[below] at (1,-0.02) {1};
    
    \node[right] at (1.05,0) {$q$};
    
    \node[above] at (0,1.25) {$u'$};
    
    \node[below] at (0.02,0) {0};
    
    \node[left] at (0,1.2) {$\frac{\beta}{\mu_n}$};

    \end{tikzpicture}
\end{figure}
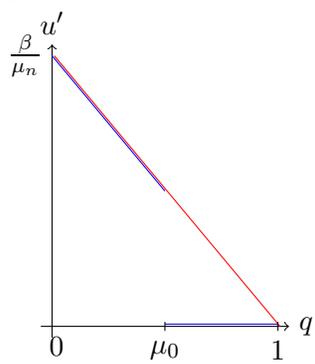
\end{proof}
\end{document}